\documentclass[journal]{IEEEtran}
\usepackage{amsmath,amssymb,bm}
\usepackage{graphicx}
\usepackage[colorlinks,citecolor=blue,linkcolor=blue]{hyperref}
\usepackage{algorithm}
\usepackage{algorithmic}

\newtheorem{remark}{Remark}
\newtheorem{proposition}{Proposition}
\newtheorem{corollary}{Corollary}
\usepackage{threeparttable}
\usepackage{multirow}

\newcommand{\E}{\mathbb{E}}
\newcommand{\C}{\mathbb{C}}
\newcommand{\R}{\mathbb{R}}

\newcommand{\tr}{\operatorname{tr}}

\begin{document}

\title{Taming the Bessel Landscape:\\Joint Antenna Position Optimization for Spatial Decorrelation in Fluid MIMO Systems}

\author{Tuo Wu, 
            Kai-Kit Wong,~\IEEEmembership{Fellow,~IEEE}, 
            Baiyang Liu, \\
            Kin-Fai Tong,~\IEEEmembership{Fellow,~IEEE}, and 
            Hyundong Shin,~\IEEEmembership{Fellow,~IEEE}
\vspace{-8mm}

\thanks{(\textit{Corresponding author: Kai-Kit Wong.})}

\thanks{This research work of T. Wu was funded by Hong Kong Research Grants Council under the Area of Excellence Scheme under Grant AoE/E-101/23-N. The work of K. K. Wong is supported by the Engineering and Physical Sciences Research Council (EPSRC) under grant EP/W026813/1. The research work of K.-F. Tong and Baiyang Liu was partly funded by Hong Kong Research Grants Council under the Area of Excellence Scheme under Grant AoE/E-101/23-N and the Hong Kong Metropolitan University, Staff Research Startup Fund: FRSF/2024/03. The work of H. Shin is supported by the National Research Foundation of Korea (NRF) grant funded by the Korean government (MSIT) (RS-2025-00556064 and RS-2025-25442355), and by the Ministry of Science and ICT (MSIT), Korea, under the ITRC (Information Technology Research Center) support program (IITP-2025-RS-2021-II212046), supervised by the IITP (Institute for Information \& Communications Technology Planning \& Evaluation).}

\thanks{T. Wu is with the Department of Electrical Engineering, City University of Hong Kong, Hong Kong, China (E-mail: $\rm tuowu2@cityu.edu.hk$).} 
\thanks{K. K. Wong is with the Department of Electronic and Electrical Engineering, University College London, WC1E 7JE, London, United Kingdom, and also with the Department of Electronic Engineering, Kyung Hee University, Yongin-si, Gyeonggi-do 17104, Republic of Korea (e-mail: $\rm kai\text{-}kit.wong@ucl.ac.uk$).}
\thanks{B. Liu and K.-F. Tong are with the School of Science and Technology, Hong Kong Metropolitan University, Hong Kong SAR, China (E-mail: $\rm \{byliu, ktong\}@hkmu.edu.hk$).} 
\thanks{H. Shin is with the Department of Electronics and Information Convergence Engineering, Kyung Hee University, Yongin-si, Gyeonggi-do 17104, Republic of Korea (e-mail: $\rm hshin@khu.ac.kr$).}
}

\markboth{}
{Wu \MakeLowercase{\textit{et al.}}: Taming the Bessel Landscape: Joint Position Optimization for FAS-MIMO Decorrelation}

\maketitle

\begin{abstract}
When the concept of fluid antenna system (FAS) is applied to multiple-input multiple-output (MIMO) systems, this gives rise to MIMO-FAS, a.k.a.~fluid MIMO. Under rich scattering, the spatial correlation matrices are governed by the zeroth-order Bessel function $J_0(\cdot)$ through the continuously adjustable antenna positions, creating a highly non-convex landscape for optimization with fluctuating local optima---the \emph{Bessel landscape}. In this paper, we tackle the joint transmitter (TX) and receiver (RX) antenna position optimization problem in fluid MIMO to maximize the ergodic capacity by shaping this landscape. Using Kronecker channel decomposition, we firstly develop a suite of analytical results that expose the problem's intrinsic structure: (i) a high signal-to-noise ratio (SNR) capacity approximation that decomposes the objective into separable log-determinant terms of the TX and RX correlation matrices, $\mathbf{R}_T$ and $\mathbf{R}_R$, respectively, (ii) a closed-form capacity loss bound linking $\det(\mathbf{R}_T)\det(\mathbf{R}_R)$ to the performance gap relative to the independent and identically distributed (i.i.d.) ideal MIMO channel, and (iii) the globally optimal inter-element spacing when the number of fluid elements at the TX is $N=2$ at the first zero of $J_0$. Guided by these insights, we propose two algorithms within an alternating optimization (AO) framework. The first algorithm is AO with particle swarm optimization (PSO) which deploys a particle swarm to explore the Bessel landscape globally without gradient information. Then in the second method, we use successive convex approximation (SCA) to obtain the gradient in closed form via $J_1(\cdot)$ to construct convex surrogates for orders-of-magnitude faster convergence. Both algorithms, AO-PSO and AO-SCA, are provably monotone convergent. Simulation results show that the proposed schemes closely approach the i.i.d.~capacity upper bound, outperforming fixed-position antenna (FPA) deployments by over $7$~bps/Hz at high SNR, while AO-SCA achieves identical performance to AO-PSO with a speedup factor exceeding $10^5$.
\end{abstract}


\vspace{-2mm}
\section{Introduction}\label{sec:intro}
\IEEEPARstart{F}{rom the first} patent by Paulraj and Kailath in 1994 \cite{Paulraj-1994} to the groundbreaking paper by Foschini in 1996 \cite{Foschini-1996}, multiple-input multiple-output (MIMO) antenna systems have since become a key enabler from third generation (3G) to now the fifth generation (5G). MIMO promises to greatly enhance communication reliability \cite{Telatar-1999,Tarokh-1998,Tarokh-1999} and spectral efficiency \cite{Raleigh-1998} through spatial diversity. Furthermore, MIMO has evolved into multiuser MIMO \cite{wong2000opt,wong2002per,wong2003jcd,Vishwanath-2003,Choi-2004,Spencer-2004}, which is the backbone of the physical layer in 5G \cite{Villalonga2022spectral}. Interest in MIMO continues unabated, and many have argued that antenna counts will keep increasing to meet growing performance demands \cite{10379539}. However, indiscriminately scaling the array size may not be the right direction in practice, as it comes with substantial hardware cost, form-factor constraints, and channel acquisition overhead. An increasing concern is the explosive increase in power consumption when the number of antennas becomes very large and the peak-to-average power ratio (PAPR) greatly compromises the efficiency of power amplifiers \cite{Hung-2014}.

Looking ahead to sixth generation (6G) \cite{Tariq-2020}, techniques that can improve MIMO without necessarily raising the number of antennas will be essential. To this end, {\em reconfigurable antennas} show great potential but their capabilities have yet to translate into system-level performance gains \cite{Bernhard-2007}. One reason is that how the available control knobs (states, patterns, tuning ranges, switching granularity) translate into repeatable and measurable gains in throughput, reliability, or coverage under realistic constraints, is not well understood. Development driven from the antenna community has often been disconnected from what mobile systems actually need, yielding impressive prototypes whose reconfigurability is difficult to exploit end-to-end or whose benefits are fragile outside idealized settings. 

To bridge this gap, Wong {\em et al.}~introduced the fluid antenna system (FAS) concept \cite{wang2021fluid,wong2020perf}. FAS is a hardware-agnostic system concept that considers the antenna as a reconfigurable physical-layer resource to broaden system design and network optimization, emphasizing shape and position flexibility \cite{new2025tut,Lu-2025,hong2025contemporary,new2025flar,wu2024flu}. Practical FAS may take forms from a range of reconfigurable antenna technologies such as movable elements \cite{Zhu-Wong-2024}, liquid-based antennas \cite{shen2024design,Shamim-2025}, reconfigurable pixels \cite{zhang2024pixel,tong-2025pixel,Wong-wc2026}, metametarial-based antennas \cite{Zhang-jsac2026,Liu-2025arxiv}, and etc. In \cite{tong2025design}, Tong {\em et al.}~discussed various implementation technologies for FAS and provided a comparative analysis.

Position reconfigurability in FAS provides a whole new way to obtain diversity for additional degree of freedom (DoF) in the spatial domain, unlike conventional fixed-position antennas (FPAs). In recent years, FAS has attracted increasing research interest, leading to a rapidly expanding body of results. For instance, efforts have been made to accurately characterize the spatial correlation among the FAS channels \cite{FAS22,khammassi2023new} and to study the diversity performance \cite{new2024fluid,new2024an}. Copula techniques have been applied successfully to analyze the performance of FAS \cite{Ghadi-2023,10678877}. In \cite{ramirez2024new}, Ram\'{i}rez-Espinosa {\em et al.}~proposed the block-correlation model (BCM) that maintains tractability and great accuracy in performance analysis for FAS. Later, \cite{LaiX} employed the model with variable correlation parameters to investigate FAS. Finding the optimal port position in FAS is an important design problem. For example, \cite{G3_chai2022SISO-FAS-PS} tackled the port selection problem when only a few channel state information (CSI) observations are available. The combination of MIMO and FAS, a.k.a.~MIMO-FAS or fluid MIMO, leads to the joint port selection and beamforming problem which was studied in \cite{XLai23,10416896,LZhu23,10767351}. FAS has also been applied to improve performance for spectrum sensing \cite{YaoJ241}, integrated sensing and communication (ISAC) \cite{wang2024fluid,zou2024shift,Zhou-isac2024,SYangWCL25}, non-orthogonal multiple access (NOMA) \cite{new-2024noma,YaoJ252,TWuTCOMM26}, and reconfigurable intelligent surfaces (RISs) \cite{Ghadi2024005,LaiX242,YaoJ251,YaoJTWC26,TWuTVT25,HChenTNSE25}. In addition, localization performance can be greatly enhanced by scalable FAS \cite{HXuTWC25,TWuJSTSP25}. Besides, the additional DoF enabled by FAS is useful for secrecy communication \cite{Security3,Security2,ghadi2024phys,Yao-2025pls,Security5,TuoW} and activity monitoring \cite{JYao2024}.

Moreover, FAS provides a new approach to multiple access, revealing that multiuser communication on the same channel use can be achieved using position reconfigurability entirely at the receiver side without the need of precoding and interference management \cite{H4_wong2022FAMA}. Fast fluid antenna multiple access (FAMA) \cite{H5_wong2023fast}, slow FAMA \cite{H6_wong2023sFAMA,Coma-2026fama,Yuan-2026fama,Dinis-2026fama,Zhang-fama2025}, compact ultra-massive array (CUMA) \cite{H12_Wong2024cuma}, coded FAMA \cite{H10_hong2024coded,H11_hong2025Downlink}, turbo FAMA \cite{Waqar-2025,Waqar-2026wcl} have recently been proposed. CSI in FAS is essential to utilize the DoF and the CSI estimation problem has been addressed in \cite{xu2024channel,zhang2024learning,new2025channel,10807122,xu2024sparse}.

Overall, FAS research is a rapidly growing area. However, despite an expanding body of literature, most existing works focus on port-selection diversity with a single radio-frequency (RF) chain and rely on simplified spatial-correlation models. For fluid MIMO systems in which multiple fluid antennas are deployed at both the transmitter and receiver, the associated design and optimization problem remains largely open, which motivates our work. Table~\ref{tab:literature} summarizes the related literature and positions our contributions relative to existing efforts. 

Specifically, fluid MIMO provides a fundamental opportunity: under rich scattering, the spatial correlation matrices---parameterized by the antenna positions via the Jakes (Clarke) zeroth-order Bessel function $J_0(\cdot)$ \cite[Section~II]{new2025tut}, \cite{Jakes74}---become explicit design variables. Their entries vary oscillatory with position, forming a non-convex \emph{Bessel landscape} that can be shaped to maximize the ergodic capacity. Realizing this opportunity, however, entails several formidable challenges. First of all, the correlation kernel $J_0(2\pi d/\lambda)$ is oscillatory and non-convex in the inter-element spacing $d$, which induces a rugged ergodic-capacity surface with numerous local optima as the antenna positions vary continuously. Second, the transmit (TX) and receive (RX) positions are coupled through the Kronecker correlation model $\mathbf{R}=\mathbf{R}_R\otimes\mathbf{R}_T$, leading to a joint optimization over $N+M$ continuous variables with an objective that does not admit a trivial decomposition in which $N$ and $M$ denote, respectively, the number of fluid antennas at the TX and RX. Third, practical implementation constraints---finite apertures and minimum spacing requirements---render the feasible set inherently non-convex, thereby precluding the direct use of standard convex-optimization machinery. Moreover, existing correlation-characterization frameworks, such as the BCM \cite{ramirez2024new,LaiX242} and variable BCM (VBCM) \cite{LaiX,TuoW}, were developed primarily for port-selection analysis and do not address multi-antenna position optimization in fluid MIMO.

\begin{table*}[]
{
\caption{Comparison of Representative FAS Works}\label{tab:literature}
\vspace{-4mm}
\begin{center}
\resizebox{.7\linewidth}{!}{
\begin{tabular}{|l|c|c|c|c|c|}
\hline
\textbf{Reference} & \textbf{System} & \textbf{FA Elements} & \textbf{Design Variable} & \textbf{Correlation Model} & \textbf{Key Analytical Results} \\
\hline
\cite{wang2021fluid,wong2020perf} & SISO & Single & Port selection & Jakes (exact) & Capacity \& outage bounds \\
\cite{FAS22} & SISO & Single & Port selection & Jakes (exact) & Correlation parameters \\
\cite{ramirez2024new, LaiX242} & SISO/RIS & Single & Port selection & BCM & Outage probability \\
\cite{LaiX, TuoW} & SISO & Single & Port selection & VBCM & Outage \& secrecy \\
\cite{new2024an} & MIMO & Multiple & Port selection & Jakes (exact) & DMT, $q$-outage capacity \\
\cite{LZhu23} & MU-MISO & Multiple & Positions (inst.\ CSI) & --- & Sum-rate optimization \\
\cite{SYangWCL25} & ISAC & Multiple & Positions (inst.\ CSI) & --- & SINR optimization \\
\hline
\textbf{This work} & \textbf{MIMO} & \textbf{Multiple} & \textbf{Positions (stat.\ CSI)} & \textbf{Jakes (exact)} & \textbf{Capacity approx., loss bound, optimal spacing} \\
\hline
\end{tabular}}
\end{center}}
\vspace{-6mm}
\end{table*}

To overcome these difficulties, we first establish a suite of analytical results that unveils the problem's intrinsic structure. In particular, a high signal-to-noise ratio (SNR) approximation of the ergodic capacity (Proposition~\ref{prop:highSNR}) shows that the objective admits an additive decomposition into log-determinant terms of the TX and RX correlation matrices, $\log_2\det(\mathbf{R}_T)$ and $\log_2\det(\mathbf{R}_R)$, thereby demonstrating a separable structure that is obscured in the original coupled formulation. Next, a closed-form capacity-loss bound (Corollary~\ref{cor:loss}) characterizes, in bps/Hz, the exact penalty incurred by spatial correlation relative to the independent and identically distributed (i.i.d.) MIMO benchmark, explicitly linking the performance gap to $\det(\mathbf{R}_T)\det(\mathbf{R}_R)$. For the special case of $N=2$, we further derive the globally optimal inter-element spacing in closed form as the first positive zero of $J_0$ (Proposition~\ref{prop:N2}), yielding a fundamental design rule: placing two antennas at approximately $0.383\lambda$ achieves full decorrelation. Collectively, these insights provide engineering intuition and directly motivate the algorithmic decomposition strategy developed next.

Building upon the separability, an alternating-optimization (AO) framework that decomposes the TX-RX position design into a sequence of single-sided subproblems is developed: at each step, one side's antenna locations are optimized while the other side is held fixed. For each subproblem, we propose two complementary solvers. The first employs particle swarm optimization (PSO), which maintains a diverse population of candidate position vectors to explore the oscillatory Bessel landscape in parallel without gradient information, thereby offering strong robustness to local optima. The second is based on successive convex approximation (SCA): by exploiting a closed-form gradient of $\log_2\det(\mathbf{R})$ expressed through the first-order Bessel function $J_1(\cdot)$, we construct a concave quadratic surrogate at each iteration and reduce the subproblem to a sequence of efficiently solvable convex programs. Both AO-PSO and AO-SCA incorporate a projection-based feasibility mechanism to enforce finite-aperture and minimum-spacing constraints, and we establish monotone convergence of the resulting objective sequence. Numerically, AO-SCA matches the capacity attained by AO-PSO while achieving orders-of-magnitude runtime savings, yielding an attractive complexity-performance tradeoff for practical deployment.


Our main contributions are summarized as follows:
\begin{itemize}
\item We formulate the ergodic capacity maximization problem for fluid MIMO with position-dependent Kronecker-structured spatial correlation under Jakes' model, subject to aperture and minimum spacing constraints.
\item We derive a high-SNR ergodic capacity approximation that decomposes the capacity into a log-determinant function of the TX and RX correlation matrices (Proposition~\ref{prop:highSNR}), establish a closed-form capacity loss bound (Corollary~\ref{cor:loss}), and obtain the optimal antenna spacing in closed form for the $N\!=\!2$ case (Proposition~\ref{prop:N2}).
\item We propose AO-PSO to solve the non-convex problem, and prove its monotone convergence. The algorithm alternately optimizes the TX and RX antenna positions using PSO with a feasibility projection mechanism.
\item We develop AO-SCA that transforms the position optimization into a sequence of convex subproblems. By deriving the gradient of $\log_2\det(\mathbf{R}_T)$ in terms of the first-order Bessel function $J_1(\cdot)$, we obtain a projected gradient ascent algorithm that is orders of magnitude faster than AO-PSO while achieving the same capacity.
\item Our Monte-Carlo simulations validate the analytical results, showing that both AO-PSO and AO-SCA significantly outperform FPA, single-sided FAS optimization, and random placement schemes, closely approaching the i.i.d.\ MIMO capacity upper bound.
\end{itemize}


\vspace{-2mm}
\section{System Model}\label{sec:system}
\subsection{Antenna Configuration}
We consider a point-to-point narrowband fluid MIMO system where both the TX and RX are equipped with FAS. The TX deploys $N$ fluid antenna elements in a linear region of length $A\lambda$, where $\lambda$ is the carrier wavelength. The antenna positions of the TX along the array axis are collected in the vector $\mathbf{t} = [t_1, t_2, \ldots, t_N]^T \in \R^N$, with $t_i \in [0, A\lambda]$ for all $i$. Similarly, the RX deploys $M$ fluid antenna elements over a linear space of $B\lambda$, with position vector $\mathbf{r} = [r_1, r_2, \ldots, r_M]^T \in \R^M$ and $r_j \in [0, B\lambda]$ for all $j$. Each fluid antenna is connected to a dedicated RF chain, and the positions $(\mathbf{t}, \mathbf{r})$ are jointly optimized to optimize system performance.

\vspace{-2mm}
\subsection{Channel Model}\label{subsec:channel}
\subsubsection{Multipath Propagation}
We adopt a geometry-based stochastic channel model with $L$ scattering paths. The physical channel matrix from the TX to RX is expressed as
\begin{equation}\label{eq:H_FRV}
\mathbf{H}(\mathbf{t},\mathbf{r})=\sqrt{\frac{C_0 d^{-\alpha}}{L}}\sum_{\ell=1}^{L} g_\ell\,\mathbf{b}(\phi_\ell, \mathbf{r})\,\mathbf{a}^H(\theta_\ell, \mathbf{t}),
\end{equation}
in which $C_0$ is the reference channel power gain at unit distance, $d$ is the link distance, $\alpha$ is the path-loss exponent, $g_\ell \sim \mathcal{CN}(0,1)$ are i.i.d.~complex path gains, and the transmit and receive field response vectors (FRVs) are
\begin{align}
  \mathbf{a}(\theta_\ell, \mathbf{t})
  &= \bigl[e^{j\frac{2\pi}{\lambda}t_1\cos\theta_\ell},
            \ldots,
            e^{j\frac{2\pi}{\lambda}t_N\cos\theta_\ell}
     \bigr]^T \!\in \C^N,
  \label{eq:FRV_tx}\\
  \mathbf{b}(\phi_\ell, \mathbf{r})
  &= \bigl[e^{j\frac{2\pi}{\lambda}r_1\cos\phi_\ell},
            \ldots,
            e^{j\frac{2\pi}{\lambda}r_M\cos\phi_\ell}
     \bigr]^T \!\in \C^M.
  \label{eq:FRV_rx}
\end{align}
Here $\theta_\ell$ and $\phi_\ell$ denote the angle of departure (AoD) and angle of arrival (AoA) of the $\ell$-th path, respectively.

\subsubsection{Spatial Correlation under Jakes Scattering}\label{subsubsec:jakes}
Under the Jakes (Clarke) isotropic scattering model \cite{Jakes74}, the AoDs $\{\theta_\ell\}$ and AoAs $\{\phi_\ell\}$ are independent and uniformly distributed over $[0, \pi]$. Taking expectations over the random angles and the i.i.d.~path gains $\{g_\ell\}$, the spatial covariance between two TX elements at positions $t_i$ and $t_{i'}$, and two RX elements at positions $r_j$ and $r_{j'}$, is computed as
\begin{multline}\label{eq:cov_product}
\E\!\left[H_{j,i}\,H_{j',i'}^*\right]= C_0 d^{-\alpha}\underbrace{\E\!\left[e^{j\frac{2\pi}{\lambda}(t_i-t_{i'})\cos\theta}\right]}_{[\mathbf{R}_T(\mathbf{t})]_{ii'}}\\
\times\underbrace{\E\!\left[e^{j\frac{2\pi}{\lambda}(r_j-r_{j'})\cos\phi}\right]}_{[\mathbf{R}_R(\mathbf{r})]_{jj'}},
\end{multline}
where $H_{j,i}$ denotes the $(j,i)$-th entry of $\mathbf{H}$. Since $\theta \sim \mathcal{U}[0,\pi]$, the expectation evaluates exactly to the zeroth-order Bessel function of the first kind \cite{FAS22}:
\begin{align}
  \E\!\left[e^{j\frac{2\pi}{\lambda}\Delta\cos\theta}\right]
  = J_0\!\left(\frac{2\pi\Delta}{\lambda}\right),
  \label{eq:bessel_identity}
\end{align}
yielding position-dependent transmit and receive spatial correlation matrices
\begin{align}
  [\mathbf{R}_T(\mathbf{t})]_{ii'}
  &= J_0\!\left(\frac{2\pi|t_i - t_{i'}|}{\lambda}\right),~i,i' = 1,\ldots,N,\label{eq:RT}\\
  [\mathbf{R}_R(\mathbf{r})]_{jj'}
  &= J_0\!\left(\frac{2\pi|r_j - r_{j'}|}{\lambda}\right),~j,j' = 1,\ldots,M.\label{eq:RR}
\end{align}
Both $\mathbf{R}_T \in \R^{N\times N}$ and $\mathbf{R}_R \in \R^{M\times M}$ are real symmetric positive semi-definite matrices whose entries are determined by the antenna position vectors $\mathbf{t}$ and $\mathbf{r}$, respectively. Note that since $J_0(\cdot)$ is real-valued, the correlation matrices are real despite originating from a complex channel model.

\begin{remark}[Properties of the Bessel Correlation Function]\label{rem:bessel}
The zeroth-order Bessel function $J_0(x)$ governs the spatial correlation in \eqref{eq:RT} and \eqref{eq:RR}, and exhibits several properties that are central to the position optimization problem:
\begin{itemize}
\item[(i)] \emph{Monotone envelope decay:}
$|J_0(x)|$ is bounded by $\sqrt{2/(\pi x)}$ for $x > 0$, implying that the correlation magnitude decreases as the inter-element spacing increases.
\item[(ii)] \emph{Oscillatory zeros:}
$J_0(x)$ has countably many positive zeros at $x_1 \approx 2.405$, $x_2 \approx 5.520$, $x_3 \approx 8.654$, \ldots, so that zero correlation is achieved at discrete spacings $d_k = x_k \lambda / (2\pi)$, $k = 1, 2, \ldots$\,. The first zero occurs at $d_1 \approx 0.383\lambda$, which is notably \emph{smaller} than the commonly used half-wavelength spacing.
\item[(iii)] \emph{Non-monotonicity:}
$J_0(x)$ takes both positive and negative values, so a larger spacing does not always yield lower correlation, making the optimization multi-modal.
\item[(iv)] \emph{Positive semi-definiteness:}
For any set of positions $\{t_i\}$, $\mathbf{R}_T$ is positive semi-definite since $J_0(\cdot)$ is the characteristic function of the uniform distribution on $[0,\pi]$ \cite{FAS22}.
\end{itemize}
\end{remark}

\subsubsection{Eigenvalue Structure of the Correlation Matrix}
Let $\mu_1 \geq \cdots \geq \mu_N \geq 0$ denote the eigenvalues of $\mathbf{R}_T$; since $\tr(\mathbf{R}_T) = N$, we have $\sum_k \mu_k = N$. Two extreme cases are of particular interest:
\begin{itemize}
\item \emph{Full rank} ($\mathbf{R}_T = \mathbf{I}_N$): All $\mu_k = 1$, corresponding to fully independent spatial channels. This is achieved when all pairwise correlations are zero.
\item \emph{Rank deficient} ($\mathbf{R}_T \to \mathbf{1}\mathbf{1}^T/N$): $\mu_1 \to N$, $\mu_k \to 0$ for $k > 1$. This occurs when all antennas are collocated, reducing the MIMO system to a single-antenna link.
\end{itemize}
The condition number $\kappa(\mathbf{R}_T) = \mu_1/\mu_N$ quantifies the eigenvalue spread: $\kappa(\mathbf{I}_N) = 1$ (optimal), while $\kappa \to \infty$ for a rank-deficient matrix. Specifically, the goal of position optimization can be viewed as to drive $\kappa$ toward unity.

\subsubsection{Kronecker Correlated Channel Model}
Based on \eqref{eq:H_FRV}--\eqref{eq:RR}, the MIMO channel admits the \emph{Kronecker decomposition}:
\begin{align}
  \mathbf{H}(\mathbf{t},\mathbf{r})
  \;=\; \mathbf{R}_R^{1/2}(\mathbf{r})\,
        \mathbf{G}\,
        \mathbf{R}_T^{1/2}(\mathbf{t}),
  \label{eq:H_Kronecker}
\end{align}
where $\mathbf{G} \in \C^{M\times N}$ has i.i.d.~entries $\mathcal{CN}(0, C_0 d^{-\alpha})$, representing the de-correlated (virtual) channel coefficients, and $\mathbf{R}_T^{1/2}$, $\mathbf{R}_R^{1/2}$ are the unique Hermitian positive semi-definite square roots of the spatial correlation matrices.

The Kronecker model \eqref{eq:H_Kronecker} \cite{new2024an} assumes that the TX and RX scattering environments are statistically independent, which is valid when the arrays are small relative to the link distance. It separates TX and RX correlation, reducing the channel statistics to the pair $(\mathbf{R}_T, \mathbf{R}_R)$.

\begin{remark}[Correlation as a Position Function]\label{rem:corr_position}
Unlike FPA systems where $\mathbf{R}_T$ and $\mathbf{R}_R$ are fixed, in FAS they are \emph{explicit optimization variables} through $(\mathbf{t}, \mathbf{r})$. This additional DoF allows the physical layer to directly shape the eigenvalue spectra of the correlation matrices, driving them toward $\mathbf{I}$ (maximum spatial DoF) by optimizing inter-antenna spacings.
\end{remark}

\vspace{-2mm}
\subsection{Signal Model and Ergodic Capacity}\label{subsec:signal}
The received signal vector at the RX is given by
\begin{align}\label{eq:received}
\mathbf{y} = \mathbf{H}(\mathbf{t},\mathbf{r})\,\mathbf{x} + \mathbf{n},
\end{align}
where $\mathbf{x} \in \C^N$ denotes the transmitted signal vector satisfying $\E[\mathbf{x}\mathbf{x}^H] = (P/N)\mathbf{I}_N$, $P$ is the total transmit power, and $\mathbf{n} \sim \mathcal{CN}(\mathbf{0}, \sigma^2\mathbf{I}_M)$ is additive white Gaussian noise. Assuming Gaussian signaling and CSI known only at the receiver (CSIR), the instantaneous mutual information is
\begin{align}
  I(\mathbf{t},\mathbf{r};\mathbf{H})
  = \log_2\!\det\!\left(
      \mathbf{I}_M + \frac{P}{N\sigma^2}
      \mathbf{H}(\mathbf{t},\mathbf{r})\,\mathbf{H}^H(\mathbf{t},\mathbf{r})
    \right).
  \label{eq:MI}
\end{align}
Substituting \eqref{eq:H_Kronecker} into \eqref{eq:MI} and taking expectations over the random channel $\mathbf{G}$, the ergodic capacity is
\begin{align}
  C(\mathbf{t},\mathbf{r})
  &= \E_{\mathbf{G}}\!\left[
       \log_2\!\det\!\left(
         \mathbf{I}_M + \frac{P}{N\sigma^2}
         \mathbf{R}_R^{1/2}\mathbf{G}\,
         \mathbf{R}_T\,
         \mathbf{G}^H \mathbf{R}_R^{1/2}
       \right)
     \right].
  \label{eq:capacity}
\end{align}
Note that $C(\mathbf{t},\mathbf{r})$ depends on $(\mathbf{t},\mathbf{r})$ exclusively through the eigenvalue distributions of $\mathbf{R}_T(\mathbf{t})$ and $\mathbf{R}_R(\mathbf{r})$, making the correlation structure the fundamental quantity to be optimized.


\begin{remark}[Monte-Carlo Capacity Estimation]\label{rem:MC}
The expectation in \eqref{eq:capacity} does not admit a closed-form expression and we estimate it via Monte-Carlo averaging:
\begin{align}
  \hat{C}(\mathbf{t},\mathbf{r})
  = \frac{1}{S}\sum_{s=1}^{S}
    \log_2\!\det\!\left(
      \mathbf{I}_M + \frac{P}{N\sigma^2}
      \mathbf{H}_s\mathbf{H}_s^H
    \right),
  \label{eq:MC_estimate}
\end{align}
where $\mathbf{H}_s = \mathbf{R}_R^{1/2}\mathbf{G}_s\mathbf{R}_T^{1/2}$ and $\{\mathbf{G}_s\}_{s=1}^S$ are i.i.d.~samples. We use $S = 200$ for optimization and $S = 1500$ for final evaluation. The \emph{same} $\{\mathbf{G}_s\}$ are shared across all position candidates within each AO iteration for fair comparison.
\end{remark}

\vspace{-2mm}
\section{Problem Formulation}\label{sec:problem}
We seek the antenna position vectors $(\mathbf{t},\mathbf{r})$ that maximize the ergodic capacity \eqref{eq:capacity} subject to practical constraints on the antenna positions. That is,
\begin{subequations}\label{eq:P0}
\begin{align}
  \max_{\mathbf{t},\,\mathbf{r}}\quad
  & C(\mathbf{t},\mathbf{r})
  \label{eq:P0_obj}\\
  \mathrm{s.t.}\quad
  & t_i \in [0,\, A\lambda], \quad \forall\, i = 1,\ldots,N,
  \label{eq:P0_trange}\\
  & r_j \in [0,\, B\lambda], \quad \forall\, j = 1,\ldots,M,
  \label{eq:P0_rrange}\\
  & |t_i - t_{i'}| \geq d_{\min}, \quad \forall\, i \neq i',
  \label{eq:P0_tspacing}\\
  & |r_j - r_{j'}| \geq d_{\min}, \quad \forall\, j \neq j',
  \label{eq:P0_rspacing}
\end{align}
\end{subequations}
where \eqref{eq:P0_trange} (resp.\ \eqref{eq:P0_rrange}) limits each TX (resp.\ RX) antenna to the physical aperture of the array. Constraints \eqref{eq:P0_tspacing} and \eqref{eq:P0_rspacing} enforce a minimum inter-element spacing $d_{\min}$ to prevent mutual coupling and ensure mechanical feasibility.

\vspace{-2mm}
\subsection{Feasibility Analysis}
Problem~\eqref{eq:P0} is feasible if and only if
\begin{align}
A\lambda \geq (N-1)d_{\min}~~\text{and}~~B\lambda \geq (M-1)d_{\min}.
\label{eq:feasibility}
\end{align}
The feasible set for the TX positions can be expressed as
\begin{align}
  \mathcal{T} \triangleq \bigl\{
    \mathbf{t} \in [0, A\lambda]^N :
    t_1 < \cdots < t_N,\;
    t_{i+1} - t_i \geq d_{\min}
  \bigr\},
\label{eq:feasible_set}
\end{align}
which is a compact polytope in $\R^N$, guaranteeing the existence of a global maximizer.

\vspace{-2mm}
\subsection{Properties of the Optimization Problem}
\begin{remark}[Non-Convexity]\label{rem:nonconvex}
Problem \eqref{eq:P0} is non-convex because (i) $J_0(\cdot)$ is oscillatory, creating multiple local optima, and (ii) the spacing constraints define a non-convex feasible set (a union of disjoint convex regions, one per antenna ordering).
\end{remark}

\begin{remark}[Capacity Bounds]\label{rem:bounds}
The ergodic capacity $C(\mathbf{t},\mathbf{r})$ is bounded by
\begin{align}
  C_{\rm FPA} \leq C^* \leq C_{\rm iid},
  \label{eq:bounds}
\end{align}
where $C^* = \max_{(\mathbf{t},\mathbf{r})\in\mathcal{T}\times\mathcal{R}}C(\mathbf{t},\mathbf{r})$ is the globally optimal capacity, $C_{\rm FPA}$ is the capacity with fixed uniform $d_{\min}$-spaced arrays (a feasible point of~\eqref{eq:P0}), and $C_{\rm iid} = \E[\log_2\det(\mathbf{I}_M + \frac{P}{N\sigma^2} \mathbf{G}\mathbf{G}^H)]$ is the i.i.d.\ MIMO capacity upper bound achieved when $\mathbf{R}_T = \mathbf{R}_R = \mathbf{I}$.
\end{remark}

\begin{remark}[Decoupled TX/RX Structure]\label{rem:decoupled}
The objective \eqref{eq:capacity} depends on $\mathbf{t}$ and $\mathbf{r}$ only through $\mathbf{R}_T(\mathbf{t})$ and $\mathbf{R}_R(\mathbf{r})$, respectively. At high SNR, the capacity approaches $\sum_{k=1}^N \E[\log_2(1 + \frac{P}{N\sigma^2}\mu_k\nu_k)]$, where $\{\mu_k\}$ and $\{\nu_k\}$ are the eigenvalues of $\mathbf{R}_T$ and $\mathbf{R}_R$. This separable structure motivates the AO framework to be developed in Section~\ref{sec:solution}. When $N = M$ and $A = B$, symmetry implies $\mathbf{t}^* = \mathbf{r}^*$.
\end{remark}

\vspace{-2mm}
\section{Performance Analysis}\label{sec:analysis}
In this section, we derive analytical expressions that characterize the relationship between the antenna positions and the ergodic capacity, providing both design insights and verifiable benchmarks for the proposed optimization algorithm.

\vspace{-2mm}
\subsection{High-SNR Capacity Approximation}
We first present an asymptotic capacity expression valid in the high-SNR regime, which reveals the explicit dependence of the ergodic capacity on the spatial correlation structure.

\begin{proposition}[High-SNR Capacity Approximation]\label{prop:highSNR}
For the fluid MIMO system with $M = N$ and $\mathbf{G}$ having i.i.d. $\mathcal{CN}(0,1)$ entries, as $\gamma \triangleq P/(N\sigma^2) \to \infty$, the ergodic capacity admits the approximation
\begin{multline}\label{eq:C_highSNR}
C(\mathbf{t},\mathbf{r})
\approx N\log_2\gamma
  + \log_2\!\det\!\bigl(\mathbf{R}_T(\mathbf{t})\bigr)\\
  + \log_2\!\det\!\bigl(\mathbf{R}_R(\mathbf{r})\bigr)+ \kappa_N,
\end{multline}
where the constant $\kappa_N$, independent of $(\mathbf{t},\mathbf{r})$, is given by
\begin{align}
  \kappa_N
  = \frac{1}{\ln 2}\sum_{m=1}^{N}\psi(m),
  \label{eq:kappa}
\end{align}
and $\psi(\cdot)$ denotes the digamma function.
\end{proposition}

\begin{IEEEproof}
At high SNR ($\gamma \gg 1$), the identity matrix in \eqref{eq:capacity}
becomes negligible, yielding
\begin{align}
  C &\approx \E\!\left[
    \log_2\!\det\!\left(
      \gamma\,\mathbf{R}_R^{1/2}\mathbf{G}\,\mathbf{R}_T\,
      \mathbf{G}^H\mathbf{R}_R^{1/2}
    \right)
  \right].
  \label{eq:proof_step1}
\end{align}
Since $M = N$, the determinant of the product factors as
\begin{multline}\label{eq:proof_step2}
\det\!\left(\gamma\,\mathbf{R}_R^{1/2}\mathbf{G}\,\mathbf{R}_T\,
  \mathbf{G}^H\mathbf{R}_R^{1/2}\right)\\
= \gamma^N\det(\mathbf{R}_R)\det(\mathbf{R}_T)\left|\det(\mathbf{G})\right|^2.
\end{multline}
Taking $\log_2(\cdot)$ and the expectation, the first three terms are deterministic. The remaining term $\E[\log_2|\det(\mathbf{G})|^2] = \E[\log_2\det(\mathbf{G}\mathbf{G}^H)]$ involves a complex Wishart matrix $\mathbf{W} = \mathbf{G}\mathbf{G}^H\sim \mathcal{CW}_N(N,\mathbf{I}_N)$, for which \cite{new2024an}
\begin{align}\label{eq:wishart_logdet}
  \E[\ln\det(\mathbf{W})]  = \sum_{m=1}^{N}\psi(m).
\end{align}
Combining and converting to $\log_2$ via the factor $1/\ln 2$ yields \eqref{eq:C_highSNR} and \eqref{eq:kappa}, which completes the proof.
\end{IEEEproof}

Proposition~\ref{prop:highSNR} reveals that at high SNR, the ergodic capacity is determined by two independent terms: $\log_2\det(\mathbf{R}_T)$ and $\log_2\det(\mathbf{R}_R)$. This motivates the following characterization of the capacity penalty incurred by spatial correlation.

\vspace{-2mm}
\subsection{Capacity Loss Due to Spatial Correlation}
\begin{corollary}[Capacity Loss]\label{cor:loss}
Define the \emph{correlation-induced capacity loss} as $\Delta C \triangleq C_{\rm iid} - C(\mathbf{t},\mathbf{r})$, where $C_{\rm iid}$ denotes the i.i.d.~MIMO capacity ($\mathbf{R}_T = \mathbf{R}_R = \mathbf{I}$). At high SNR,
\begin{align}
  \Delta C
  \approx -\log_2\!\det\!\bigl(\mathbf{R}_T(\mathbf{t})\bigr)
         -\log_2\!\det\!\bigl(\mathbf{R}_R(\mathbf{r})\bigr)
  \geq 0,
  \label{eq:DeltaC}
\end{align}
with equality if and only if $\mathbf{R}_T = \mathbf{R}_R = \mathbf{I}_N$.
\end{corollary}

\begin{IEEEproof}
For the i.i.d.~case, $\det(\mathbf{R}_T) = \det(\mathbf{R}_R) = 1$; hence, $C_{\rm iid} \approx N\log_2\gamma + \kappa_N$. Subtracting \eqref{eq:C_highSNR} gives \eqref{eq:DeltaC}. The non-negativity follows from Hadamard's inequality: for any positive semi-definite $\mathbf{R}$ with unit diagonal, $\det(\mathbf{R}) \leq \prod_{i=1}^{N}[\mathbf{R}]_{ii} = 1$; thus, $\log_2\det(\mathbf{R}) \leq 0$. Equality holds if and only if $\mathbf{R}$ is diagonal, i.e., $\mathbf{R} = \mathbf{I}$.
\end{IEEEproof}

\begin{remark}[Determinant Maximization Interpretation]
Corollary~\ref{cor:loss} establishes that in the high SNR regime, \eqref{eq:P0} is asymptotically equivalent to
\begin{align}
  \max_{\mathbf{t},\mathbf{r}} \;
  \det\!\bigl(\mathbf{R}_T(\mathbf{t})\bigr)
  \det\!\bigl(\mathbf{R}_R(\mathbf{r})\bigr),
  \label{eq:det_max}
\end{align}
subject to the constraints \eqref{eq:P0_trange}--\eqref{eq:P0_rspacing}. Since the determinant of a correlation matrix equals the product of its eigenvalues, the optimal positions should spread the eigenvalue spectrum as uniformly as possible, driving $\mathbf{R}_T$ and $\mathbf{R}_R$ toward the identity.
\end{remark}

\vspace{-2mm}
\subsection{Optimal Spacing Analysis for $N = 2$}
For the special case $N = 2$, we can derive the optimal inter-element spacing in closed form.

\begin{proposition}[Optimal Spacing for $N = 2$]\label{prop:N2}
Consider $N = M = 2$ with aperture $A\lambda \geq d_{\min}$. The TX correlation matrix can be written as
\begin{align}
  \mathbf{R}_T =
  \begin{bmatrix}
    1 & J_0\!\left(\frac{2\pi d}{\lambda}\right) \\[3pt]
    J_0\!\left(\frac{2\pi d}{\lambda}\right) & 1
  \end{bmatrix},
  \label{eq:RT_N2}
\end{align}
where $d = |t_1 - t_2|$ is the antenna spacing. Its determinant is
\begin{align}
  \det(\mathbf{R}_T) = 1 - J_0^2\!\left(\frac{2\pi d}{\lambda}\right).
  \label{eq:det_N2}
\end{align}
The determinant is maximized (i.e., $\det(\mathbf{R}_T) = 1$ and $\mathbf{R}_T = \mathbf{I}_2$) at the optimal spacing
\begin{align}
  d^* = \frac{x_1\lambda}{2\pi} \approx 0.383\lambda,
  \label{eq:dstar}
\end{align}
where $x_1 \approx 2.4048$ is the first positive zero of $J_0(x)$.
\end{proposition}

\begin{IEEEproof}
We have $\det(\mathbf{R}_T) = 1 - \rho^2(d)$ where $\rho(d) = J_0(2\pi d/\lambda)$. Since $|\rho(d)| \leq 1$, $\det(\mathbf{R}_T) \leq 1$ with equality if and only if $\rho(d) = 0$, i.e., $J_0(2\pi d/\lambda) = 0$. The first positive zero of $J_0(x)$ occurs at $x_1 \approx 2.4048$, giving $d^* = x_1\lambda/(2\pi) \approx 0.383\lambda$. At this spacing, $\mathbf{R}_T = \mathbf{I}_2$, achieving full decorrelation and eliminating the capacity loss.
\end{IEEEproof}

\begin{remark}[Beyond $N = 2$]\label{rem:beyond_N2}
For $N > 2$, achieving $\mathbf{R}_T = \mathbf{I}_N$ requires $J_0(2\pi|t_i - t_{i'}|/\lambda) = 0$ for all $i \neq i'$. With uniform spacing $d$, the pairwise distances are $d, 2d, 3d, \ldots, (N-1)d$. While $d^* = 0.383\lambda$ zeroes out $J_0$ for distance $d$, the Bessel function $J_0(2\pi \cdot 2d^*/\lambda)= J_0(2x_1) \approx -0.178 \neq 0$. Therefore, uniform spacing at $d^*$ does not achieve $\mathbf{R}_T = \mathbf{I}_N$. This fundamental limitation motivates the AO-PSO to be developed in the next section, which can optimize non-uniform antenna spacings to maximize $\det(\mathbf{R}_T)$.
\end{remark}

\vspace{-2mm}
\subsection{Low-SNR Regime Analysis}
The high-SNR approximation in Proposition~\ref{prop:highSNR} becomes inaccurate when $\gamma = P/(N\sigma^2)$ is small. In the low-SNR regime, the capacity exhibits fundamentally different behavior with respect to the antenna positions.

\begin{proposition}[Low-SNR Capacity Approximation]\label{prop:lowSNR}
At low SNR ($\gamma \to 0$), the ergodic capacity satisfies
\begin{align}
  C(\mathbf{t},\mathbf{r})
  \approx \frac{\gamma}{\ln 2}
  \tr\bigl(\mathbf{R}_T(\mathbf{t})\bigr)
  \tr\bigl(\mathbf{R}_R(\mathbf{r})\bigr)
  = \frac{NM\gamma}{\ln 2},
  \label{eq:C_lowSNR}
\end{align}
which is \emph{independent} of the antenna positions.
\end{proposition}

\begin{IEEEproof}
Using the first-order Taylor expansion $\log_2\det(\mathbf{I} + \gamma\mathbf{A}) \approx \frac{\gamma}{\ln 2}\tr(\mathbf{A})$ for $\gamma \to 0$,
and substituting into~\eqref{eq:capacity} gives
\begin{multline}\label{eq:lowSNR_proof}
  C\approx \frac{\gamma}{\ln 2}
  \E\bigl[\tr\bigl(\mathbf{R}_R^{1/2}\mathbf{G}\,\mathbf{R}_T\,
  \mathbf{G}^H\mathbf{R}_R^{1/2}\bigr)\bigr]\\
= \frac{\gamma}{\ln 2}\tr\bigl(\mathbf{R}_R\bigr)\tr\bigl(\mathbf{R}_T\bigr),
\end{multline}
where the last equality follows from $\E[\mathbf{G}\,\mathbf{R}_T\,\mathbf{G}^H] = \tr(\mathbf{R}_T)\mathbf{I}_M$ for i.i.d.~$\mathbf{G}$. Since $\tr(\mathbf{R}_T) = N$ and $\tr(\mathbf{R}_R) = M$ regardless of the antenna positions, the result follows.
\end{IEEEproof}

\begin{remark}[SNR-Dependent Benefit of FAS]
\label{rem:snr_benefit}
Proposition~\ref{prop:lowSNR} reveals a fundamental dichotomy: at low SNR, the capacity is determined solely by the total received power ($NM\gamma$) and is insensitive to the spatial correlation structure, so all placement schemes perform similarly. At high SNR, the multiplexing gain dominates and depends on $\det(\mathbf{R}_T)\cdot\det(\mathbf{R}_R)$ (Proposition~\ref{prop:highSNR}). Thus, FAS position optimization provides the greatest benefit in the moderate-to-high SNR regime, where the correlation-induced capacity loss $\Delta C$ in~\eqref{eq:DeltaC} is most significant.
\end{remark}

\vspace{-2mm}
\subsection{Capacity Scaling with Aperture Size}
\begin{proposition}[Minimum Aperture for Full Decorrelation ($N = 2$)]\label{prop:aperture}
For $N = 2$, the minimum aperture for full decorrelation is $A_{\min}\lambda = d^* \approx 0.383\lambda$. For $N > 2$, full decorrelation is generally unattainable, but $\Delta C$ decreases monotonically with increasing aperture since $\mathcal{T}(A_1) \supset \mathcal{T}(A_0)$ for $A_1 > A_0$, and $|J_0(x)| \leq \sqrt{2/(\pi x)}$ ensures $\det(\mathbf{R}_T) \to 1$ as $A \to \infty$. In practice, $A = 2\lambda$--$3\lambda$ suffices for moderate $N$.
\end{proposition}

\vspace{-2mm}
\section{Proposed AO-PSO Algorithm}\label{sec:solution}
We propose an AO framework that decomposes \eqref{eq:P0} into two subproblems, each solved by a PSO subroutine. PSO is a derivative-free global optimizer that handles the Monte-Carlo-estimated, non-differentiable fitness landscape with bound and spacing constraints via a lightweight projection step.

\vspace{-2mm}
\subsection{AO Framework}
Starting from a feasible initial point, AO alternates:

\textbf{Step 1 (TX update):}
With $\mathbf{r}$ held fixed, solve
\begin{align}\label{eq:Ptx}
  \mathbf{t}^{(k+1)} = \arg\max_{\mathbf{t}} \; C(\mathbf{t},\mathbf{r}^{(k)})~~\mathrm{s.t.}~~\eqref{eq:P0_trange},\,\eqref{eq:P0_tspacing}.
\end{align}

\textbf{Step 2 (RX update):}
With $\mathbf{t}^{(k+1)}$ held fixed, solve 
\begin{align}\label{eq:Prx}
  \mathbf{r}^{(k+1)} = \arg\max_{\mathbf{r}} \; C(\mathbf{t}^{(k+1)},\mathbf{r})~~\mathrm{s.t.}~~\eqref{eq:P0_rrange},\,\eqref{eq:P0_rspacing}.
\end{align}

\begin{algorithm}[t]
\caption{AO-PSO for Joint FAS Position Optimization}\label{alg:AO}
{\small 
\begin{algorithmic}[1]
\REQUIRE Apertures $A,B$; antenna counts $N,M$; power $P$;
         spacing $d_{\min}$; Monte-Carlo samples $S$;
         PSO parameters $(Z, I_{\rm PSO}, w_{\max}, w_{\min}, c_1, c_2)$;
         AO tolerance $\varepsilon$.
\ENSURE  Optimized positions $\mathbf{t}^*, \mathbf{r}^*$.
\STATE Initialize $\mathbf{t}^{(0)},\mathbf{r}^{(0)}$ by uniform spacing
       in $[0,A\lambda]$ and $[0,B\lambda]$, respectively.
\STATE Set $k \leftarrow 0$,
       $C^{(-1)} \leftarrow 0$.
\REPEAT
  \STATE $\mathbf{t}^{(k+1)} \leftarrow
         \textsc{PSO}\bigl(\mathbf{r}^{(k)},\,\text{TX side}\bigr)$
         \hfill $\triangleright$ solve \eqref{eq:Ptx}
  \STATE $\mathbf{r}^{(k+1)} \leftarrow
         \textsc{PSO}\bigl(\mathbf{t}^{(k+1)},\,\text{RX side}\bigr)$
         \hfill $\triangleright$ solve \eqref{eq:Prx}
  \STATE $C^{(k)} \leftarrow C\bigl(\mathbf{t}^{(k+1)},\mathbf{r}^{(k+1)}\bigr)$
         via \eqref{eq:capacity_MC}
  \STATE $k \leftarrow k+1$
\UNTIL{$\bigl|C^{(k)} - C^{(k-1)}\bigr| \leq \varepsilon$}
\STATE $\mathbf{t}^* \leftarrow \mathbf{t}^{(k)}$,\;
       $\mathbf{r}^* \leftarrow \mathbf{r}^{(k)}$.
\end{algorithmic}}
\end{algorithm}

\vspace{-2mm}
\subsection{PSO Subroutine for Each Subproblem}\label{subsec:pso}
Here, we describe the PSO subroutine for the TX subproblem \eqref{eq:Ptx}; the RX subproblem is handled identically with $\mathbf{t} \leftrightarrow
\mathbf{r}$, $N \leftrightarrow M$, and $A \leftrightarrow B$.

\subsubsection{Particle Representation and Initialization}
Each particle $z \in \{1,\ldots,Z\}$ represents a candidate position vector $\mathbf{t}_z \in \R^N$.  We initialize positions by drawing $N$ order statistics of a uniform distribution over $[0, A\lambda]$ and then sorting them to automatically satisfy the spacing constraint:
\begin{align}
  t_{z,1}^{(0)} < t_{z,2}^{(0)} < \cdots < t_{z,N}^{(0)}, \quad
  t_{z,i+1}^{(0)} - t_{z,i}^{(0)} \geq d_{\min}.\nonumber
\end{align}
All particle velocities are initialized to zero: $\mathbf{v}_{z}^{(0)} = \mathbf{0}$.

\subsubsection{Fitness Function Evaluation}
The fitness of particle $z$ at iteration $\ell$ is the estimated ergodic capacity with the fixed RX correlation matrix $\mathbf{R}_R$ (from the current AO iterate):
\begin{align}
  f\bigl(\mathbf{t}_z^{(\ell)}\bigr)
  &= \frac{1}{S}\sum_{s=1}^{S}
    \log_2\!\det\!\Bigl(
      \mathbf{I}_M + \tfrac{P}{N\sigma^2}
      \mathbf{R}_R^{\frac{1}{2}}
      \mathbf{G}_s
      \mathbf{R}_T\!\bigl(\mathbf{t}_z^{(\ell)}\bigr)
      \mathbf{G}_s^H
      \mathbf{R}_R^{\frac{1}{2}}
    \Bigr),
  \label{eq:capacity_MC}
\end{align}
where $\{\mathbf{G}_s\}_{s=1}^S$ are the same i.i.d.~$\mathcal{CN}(\mathbf{0},\,C_0 d^{-\alpha}\mathbf{I})$ realizations shared across all particles in the same AO iteration.

\subsubsection{Velocity and Position Update}
At each PSO iteration $\ell$, the personal best and global best are updated as
\begin{align}
  \mathbf{p}_z^{(\ell)}
  &= \arg\max\bigl\{f(\mathbf{t}_z^{(0)}),\ldots,f(\mathbf{t}_z^{(\ell)})\bigr\},
  \label{eq:pbest}\\
  \mathbf{g}^{(\ell)}
  &= \arg\max_z \; f\bigl(\mathbf{p}_z^{(\ell)}\bigr).
  \label{eq:gbest}
\end{align}
The velocity of particle $z$ is updated with a linearly decaying inertia weight $w^{(\ell)}$:
\begin{align}
  w^{(\ell)} &= w_{\max} - \frac{(w_{\max}-w_{\min})\ell}{I_{\rm PSO}},
  \label{eq:inertia}\\
  \mathbf{v}_z^{(\ell+1)}
  &= w^{(\ell)}\mathbf{v}_z^{(\ell)}
   + c_1 \varepsilon_1 \!\odot\!(\mathbf{p}_z^{(\ell)} - \mathbf{t}_z^{(\ell)})
   \notag\\
  &\quad
   + c_2 \varepsilon_2 \!\odot\!(\mathbf{g}^{(\ell)} - \mathbf{t}_z^{(\ell)}),
  \label{eq:vel_update}
\end{align}
where $c_1$ and $c_2$ are the cognitive and social learning factors, $\varepsilon_1, \varepsilon_2 \sim \mathcal{U}[0,1]^N$ are element-wise random vectors, and $\odot$ denotes the Hadamard product. The position is updated as
\begin{align}
  \tilde{\mathbf{t}}_z^{(\ell+1)}
  = \mathbf{t}_z^{(\ell)} + \mathbf{v}_z^{(\ell+1)}.
  \label{eq:pos_update}
\end{align}

\subsubsection{Feasibility Projection}
To restore feasibility after the position update, we apply the projection operator $\mathcal{P}(\cdot)$ in two stages:
\begin{itemize}
  \item \textit{Bound clipping}: clip each coordinate to $[0,A\lambda]$:
  \begin{align}
    [\mathcal{P}_1(\tilde{\mathbf{t}})]_i
    = \min\!\bigl(A\lambda,\,\max(0,\tilde{t}_i)\bigr).
    \label{eq:clip}
  \end{align}
  \item \textit{Spacing enforcement}: sort the clipped vector in ascending order, and then iteratively push overlapping antennas apart until
  $t_{i+1} - t_i \geq d_{\min}$ for all $i$:
  \begin{align}
    \mathbf{t}_z^{(\ell+1)} = \mathcal{P}_2\!\bigl(\mathcal{P}_1(\tilde{\mathbf{t}}_z^{(\ell+1)})\bigr).
    \label{eq:proj}
  \end{align}
\end{itemize}
The PSO subroutine is summarized in Algorithm~\ref{alg:PSO}.

\begin{algorithm}[t]
\caption{\textsc{PSO} Subroutine (TX side)}\label{alg:PSO}
{\small 
\begin{algorithmic}[1]
\REQUIRE Fixed $\mathbf{R}_R$; channel samples $\{\mathbf{G}_s\}$;
         parameters $Z, I_{\rm PSO}, w_{\max}, w_{\min}, c_1, c_2$.
\ENSURE  Optimized TX positions $\mathbf{t}^*$.
\STATE Initialize $\{\mathbf{t}_z^{(0)}\}$ by sorted uniform sampling;
       $\mathbf{v}_z^{(0)} \leftarrow \mathbf{0}$,\;
       $\mathbf{p}_z \leftarrow \mathbf{t}_z^{(0)}$,\;
       $\mathbf{g} \leftarrow \arg\max_z f(\mathbf{t}_z^{(0)})$.
\FOR{$\ell = 0, 1, \ldots, I_{\rm PSO}-1$}
  \FOR{$z = 1$ \TO $Z$}
    \STATE Draw $\varepsilon_1, \varepsilon_2 \sim \mathcal{U}[0,1]^N$.
    \STATE Update velocity via \eqref{eq:vel_update}.
    \STATE Update position via \eqref{eq:pos_update}.
    \STATE Project: $\mathbf{t}_z^{(\ell+1)} \leftarrow
           \mathcal{P}_2(\mathcal{P}_1(\tilde{\mathbf{t}}_z^{(\ell+1)}))$.
    \STATE Evaluate $f(\mathbf{t}_z^{(\ell+1)})$ via \eqref{eq:capacity_MC}.
    \IF{$f(\mathbf{t}_z^{(\ell+1)}) > f(\mathbf{p}_z)$}
      \STATE $\mathbf{p}_z \leftarrow \mathbf{t}_z^{(\ell+1)}$.
    \ENDIF
  \ENDFOR
  \STATE $\mathbf{g} \leftarrow \arg\max_z f(\mathbf{p}_z)$.
\ENDFOR
\STATE $\mathbf{t}^* \leftarrow \mathbf{g}$.
\end{algorithmic}}
\end{algorithm}

\vspace{-2mm}
\subsection{Convergence and Complexity Analysis}\label{subsec:complexity}
\subsubsection{Convergence}
\begin{proposition}[Monotone Convergence]\label{prop:convergence}
The sequence $\{C(\mathbf{t}^{(k)},\mathbf{r}^{(k)})\}_{k\geq 0}$ generated by Algorithm~\ref{alg:AO} is monotonically non-decreasing and bounded above by the i.i.d.~MIMO capacity $C_{\rm iid}$; hence it converges to a limit $C^\infty \leq C_{\rm iid}$.
\end{proposition}

\begin{IEEEproof}
Consider the TX update step at iteration $k$. Since the current iterate $\mathbf{t}^{(k)}$ is included in the PSO swarm (as a warm-started particle), the PSO global best satisfies
\begin{align}
  C\bigl(\mathbf{t}^{(k+1)},\mathbf{r}^{(k)}\bigr)
  \geq C\bigl(\mathbf{t}^{(k)},\mathbf{r}^{(k)}\bigr).
  \label{eq:mono_tx}
\end{align}
By the same argument for the RX update:
\begin{align}
  C\bigl(\mathbf{t}^{(k+1)},\mathbf{r}^{(k+1)}\bigr)
  \geq C\bigl(\mathbf{t}^{(k+1)},\mathbf{r}^{(k)}\bigr).
  \label{eq:mono_rx}
\end{align}
Combining the results, we have $C^{(k+1)} \geq C^{(k)}$. The upper bound follows from Remark~\ref{rem:bounds}: $C(\mathbf{t},\mathbf{r}) \leq C_{\rm iid}$ for any $(\mathbf{t},\mathbf{r})$. Since $\{C^{(k)}\}$ is non-decreasing and bounded above, it converges by the monotone convergence theorem.
\end{IEEEproof}

\begin{remark}[Convergence to Local Optimum]\label{rem:local_opt}
Proposition~\ref{prop:convergence} guarantees convergence of the objective sequence, but not to a \emph{global} optimum. The AO framework converges to a coordinate-wise stationary point. To mitigate the risk of poor local optima, we employ (i)~diversified PSO swarm initialization and (ii)~warm-starting one particle from the current AO iterate. In practice, near-identical solutions are obtained across different initializations (Section~\ref{sec:sim}), suggesting a favorable landscape with few significant local optima.
\end{remark}

\subsubsection{Computational Complexity}
Let $I_{\rm AO}$ denote the number of AO outer iterations. The dominant cost per outer iteration consists of two PSO runs (one for TX, one for RX), each evaluating $Z$ particles over $I_{\rm PSO}$ iterations. Each fitness evaluation requires computing the correlation matrix $\mathbf{R}_T$ (or $\mathbf{R}_R$), performing a Cholesky factorization, and averaging $S$ capacity samples. Specifically, we have:
\begin{itemize}
\item \emph{Correlation matrix computation:}
$\mathbf{R}_T(\mathbf{t})$ requires $N(N-1)/2$ Bessel function evaluations
and an $N \times N$ Cholesky decomposition, costing $\mathcal{O}(N^3)$.
\item \emph{Per-sample capacity evaluation:}
For each Monte-Carlo sample $s$, forming $\mathbf{H}_s = \mathbf{R}_R^{1/2}\mathbf{G}_s\mathbf{R}_T^{1/2}$
costs $\mathcal{O}(MN^2 + M^2N)$, and computing
$\det(\mathbf{I}_M + \gamma\mathbf{H}_s\mathbf{H}_s^H)$ via
Cholesky factorization costs $\mathcal{O}(M^3)$.
\item \emph{Total:}
Over $I_{\rm PSO}$ PSO iterations, $Z$ particles, $S$ Monte-Carlo samples, and two AO half-steps per outer loop, the total complexity can be estimated as
\begin{align}
  \mathcal{O}\!\left(
    I_{\rm AO}\cdot I_{\rm PSO}\cdot Z\cdot S\cdot(M^3 + N^3 + MN\max(M,N))
  \right).
  \label{eq:complexity}
\end{align}
\end{itemize}
For the default parameters used in our simulations ($I_{\rm AO}=12$, $I_{\rm PSO}=60$, $Z=20$, $S=200$, $M=N=6$), the total operation count is approximately $6.2\times 10^9$ scalar multiplications, which completes in about $3$--$5$ minutes on a modern workstation with MATLAB.

\vspace{-2mm}
\section{SCA-Based Alternative}\label{sec:sca}
AO-PSO entails significant computational overhead from Monte-Carlo fitness evaluations. In this section, we develop an alternative SCA approach that exploits the high-SNR structure to construct a sequence of tractable convex subproblems.

\vspace{-2mm}
\subsection{High-SNR Equivalent Reformulation}
From Proposition~\ref{prop:highSNR} and the determinant maximization interpretation in~\eqref{eq:det_max}, at high SNR, \eqref{eq:P0} is equivalent to
\begin{subequations}\label{eq:P_det}
\begin{align}
  \max_{\mathbf{t},\,\mathbf{r}} f(\mathbf{t}) + f(\mathbf{r})~~\mathrm{s.t.}~~\mathbf{t} \in \mathcal{T},\;
    \mathbf{r} \in \mathcal{R},
\label{eq:P_det_con}
\end{align}
\end{subequations}
where $f(\mathbf{t}) \triangleq \log_2\!\det(\mathbf{R}_T(\mathbf{t}))$. The objective is separable in $\mathbf{t}$ and $\mathbf{r}$, so each can be optimized independently. However, $f(\mathbf{t})$ is \emph{non-convex} because the mapping $\mathbf{t} \mapsto \mathbf{R}_T(\mathbf{t})$ is nonlinear through the oscillatory $J_0(\cdot)$.

\vspace{-2mm}
\subsection{Closed-Form Gradient via Bessel Identities}
We derive the gradient of $f(\mathbf{t})$ with respect to each position $t_n$. By the chain rule for matrix functions,
\begin{align}
  \frac{\partial f}{\partial t_n}
  = \frac{1}{\ln 2}\,
    \tr\!\left(
      \mathbf{R}_T^{-1}\,
      \frac{\partial \mathbf{R}_T}{\partial t_n}
    \right).
  \label{eq:grad_chain}
\end{align}
Using the Bessel derivative identity $J_0'(x) = -J_1(x)$, the off-diagonal Jacobian entries can be found as
\begin{align}
  \frac{\partial [\mathbf{R}_T]_{nj}}{\partial t_n}
  = -\frac{2\pi}{\lambda}\,
     J_1\!\left(\frac{2\pi|t_n - t_j|}{\lambda}\right)
     \operatorname{sgn}(t_n - t_j),
  \label{eq:dRdt}
\end{align}
for $j \neq n$, with $\partial[\mathbf{R}_T]_{nn}/\partial t_n = 0$. The full gradient component is given by
\begin{align}
  \frac{\partial f}{\partial t_n}
  = -\frac{4\pi}{\lambda\ln 2}
    \sum_{\substack{j=1\\j\neq n}}^{N}
    [\mathbf{R}_T^{-1}]_{nj}\,
    J_1\!\left(\frac{2\pi|t_n\!-\!t_j|}{\lambda}\right)
    \operatorname{sgn}(t_n\!-\!t_j).
  \label{eq:grad_fn}
\end{align}
Computing $\nabla_{\mathbf{t}} f \in \R^N$ requires one $N \times N$ matrix inversion and $N(N-1)/2$ evaluations of $J_1(\cdot)$, costing $\mathcal{O}(N^3)$---with \emph{no Monte-Carlo sampling}.

\begin{remark}[Gradient Interpretation]
The term $[\mathbf{R}_T^{-1}]_{nj}$ captures the sensitivity of $\log\det(\mathbf{R}_T)$ to the $(n,j)$-th correlation entry. Since $J_1(x) > 0$ for small $x$, the gradient pushes closely spaced antennas apart to reduce correlation. At larger spacings, it fine-tunes positions to hit the zeros of $J_0$.
\end{remark}

\vspace{-2mm}
\subsection{Projected Gradient Ascent and SCA Framework}
SCA constructs a first-order surrogate at iterate $\mathbf{t}^{(k)}$:
\begin{align}
  \hat{f}(\mathbf{t};\mathbf{t}^{(k)})
  = f(\mathbf{t}^{(k)})
  + \nabla f(\mathbf{t}^{(k)})^T (\mathbf{t} - \mathbf{t}^{(k)}),
  \label{eq:surrogate}
\end{align}
which is affine (hence convex) in $\mathbf{t}$. Rather than solving the resulting linear programming on $\mathcal{T}$ (which yields extreme-point solutions and may oscillate), we use a \emph{proximal projected gradient ascent} step:
\begin{align}
  \mathbf{t}^{(k+1)}
  = \mathcal{P}_{\mathcal{T}}\!\left(
      \mathbf{t}^{(k)} + \eta\,\nabla f(\mathbf{t}^{(k)})
    \right),
  \label{eq:PGA}
\end{align}
where $\eta > 0$ is the step size and $\mathcal{P}_{\mathcal{T}}$ is the projection onto $\mathcal{T}$ (bound clipping + spacing enforcement, identical to~\eqref{eq:clip} and \eqref{eq:proj}). A backtracking line search on $\eta$ ensures $f(\mathbf{t}^{(k+1)}) \geq f(\mathbf{t}^{(k)})$. AO-SCA (Algorithm~\ref{alg:AO_SCA}) alternates between TX and RX updates, each solved by $I_{\rm SCA}$ projected gradient ascent steps on the log-determinant objective.

\begin{algorithm}[t]
\caption{AO-SCA for Joint FAS Position Optimization}\label{alg:AO_SCA}
{\small
\begin{algorithmic}[1]
\REQUIRE Apertures $A,B$; counts $N,M$; spacing $d_{\min}$;
         step size $\eta_0$; SCA iterations $I_{\rm SCA}$;
         AO iterations $I_{\rm AO}$; tolerance $\varepsilon$.
\ENSURE  Optimized positions $\mathbf{t}^*, \mathbf{r}^*$.
\STATE Initialize $\mathbf{t}^{(0)}, \mathbf{r}^{(0)}$ by uniform spacing.
\FOR{$k = 0, 1, \ldots, I_{\rm AO}-1$}
  \STATE \textbf{TX update:} Fix $\mathbf{r}^{(k)}$.
  \FOR{$\ell = 0, 1, \ldots, I_{\rm SCA}-1$}
    \STATE Compute $\mathbf{R}_T(\mathbf{t}^{(\ell)})$ and
           $\mathbf{R}_T^{-1}$ via Cholesky.
    \STATE Compute $\nabla_{\mathbf{t}} f$ via~\eqref{eq:grad_fn}.
    \STATE Backtracking: find $\eta$ s.t.\
           $f(\mathcal{P}_{\mathcal{T}}(\mathbf{t}^{(\ell)} + \eta\nabla f))
           > f(\mathbf{t}^{(\ell)})$.
    \STATE $\mathbf{t}^{(\ell+1)} \leftarrow
           \mathcal{P}_{\mathcal{T}}(\mathbf{t}^{(\ell)} + \eta\nabla f)$.
  \ENDFOR
  \STATE $\mathbf{t}^{(k+1)} \leftarrow \mathbf{t}^{(I_{\rm SCA})}$.
  \STATE \textbf{RX update:} Repeat above for $\mathbf{r}$ with
         $f(\mathbf{r}) = \log_2\det(\mathbf{R}_R(\mathbf{r}))$.
\ENDFOR
\end{algorithmic}}
\end{algorithm}

\vspace{-2mm}
\subsection{Convergence and Complexity Comparison}
\begin{proposition}[AO-SCA Convergence]\label{prop:sca_conv}
With backtracking line search, the AO-SCA algorithm generates a monotonically non-decreasing sequence $\{f(\mathbf{t}^{(k)}) + f(\mathbf{r}^{(k)})\}$ that converges to a coordinate-wise stationary point of~\eqref{eq:P_det}.
\end{proposition}

\begin{IEEEproof}
Each projected gradient step with backtracking guarantees $f(\mathbf{t}^{(\ell+1)}) \geq f(\mathbf{t}^{(\ell)})$. Over the $I_{\rm SCA}$ inner iterations, $f(\mathbf{t}^{(k+1)}) \geq f(\mathbf{t}^{(k)})$. The same holds for the RX update. Since $f(\mathbf{t}) \leq 0$ for any $\mathbf{t}$ (Hadamard's inequality), the sequence is bounded above and hence converges.
\end{IEEEproof}

The per-iteration complexity of AO-SCA is estimated as
\begin{align}
  \mathcal{O}(I_{\rm SCA} \cdot N^3),
  \label{eq:sca_complexity}
\end{align}
which is \emph{independent of the Monte-Carlo sample size $S$}. In contrast, AO-PSO costs $\mathcal{O}(I_{\rm PSO} \cdot Z \cdot S \cdot N^3)$ per AO step. For typical parameters ($I_{\rm PSO} = 60$, $Z = 20$, $S = 200$), AO-SCA is approximately $2.4 \times 10^5$ times faster per AO iteration.

\begin{remark}[SCA vs.\ PSO Trade-off]\label{rem:sca_vs_pso}
AO-SCA is a \emph{local} optimizer: it converges to the nearest stationary point from the given initialization. AO-PSO is a \emph{global} optimizer: the stochastic swarm explores the search space broadly and is less sensitive to initialization. Moreover, AO-SCA optimizes the \emph{high-SNR surrogate} $\log_2\det(\mathbf{R}_T)$ rather than the exact ergodic capacity, so its solution may be suboptimal at moderate SNR. In practice, as will be shown in Section~\ref{sec:sim}, AO-SCA achieves near-identical capacity to AO-PSO at high SNR while being orders of magnitude faster, but AO-PSO is more robust at moderate SNR and with poor initialization.
\end{remark}

\vspace{-2mm}
\section{Simulation Results}\label{sec:sim}
Unless otherwise stated, the default parameters are listed in Table~\ref{tab:params}. The channel follows~\eqref{eq:H_Kronecker} with $\mathbf{G} \sim \mathcal{CN}(\mathbf{0},\mathbf{I})$.

\begin{table}[t]
\centering
\caption{Default Simulation Parameters}\label{tab:params}
\resizebox{.85\columnwidth}{!}{
\begin{tabular}{lll}
\hline
\textbf{Parameter} & \textbf{Symbol} & \textbf{Value} \\
\hline
Number of TX/RX antennas & $N, M$ & $6$ \\
TX/RX aperture & $A\lambda, B\lambda$ & $2\lambda$ \\
Minimum spacing & $d_{\min}$ & $0.3\lambda$ \\
Noise power & $\sigma^2$ & $1$ \\
PSO swarm size & $Z$ & $20$ \\
PSO iterations & $I_{\rm PSO}$ & $60$ \\
Inertia weight & $w_{\max}, w_{\min}$ & $0.9, 0.4$ \\
Learning factors & $c_1, c_2$ & $1.5, 1.5$ \\
Monte-Carlo samples (optimization) & $S$ & $200$ \\
Monte-Carlo samples (evaluation) & $S_{\rm eval}$ & $1500$ \\
AO iterations & $I_{\rm AO}$ & $12$ \\
AO convergence threshold & $\varepsilon$ & $10^{-3}$ \\
\hline
\end{tabular}}
\end{table}

\vspace{-2mm}
\subsection{Convergence Behavior}

\begin{figure}[t]
\centering
\includegraphics[width=.85\columnwidth]{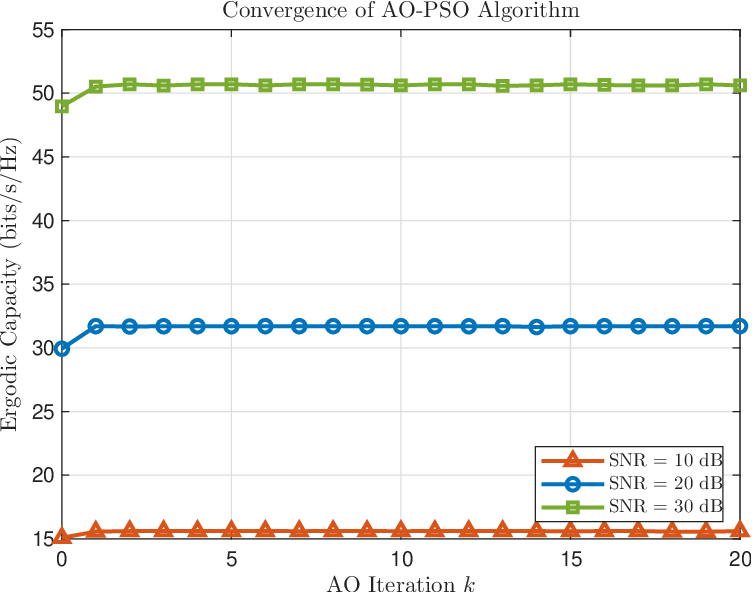}
\caption{Convergence of the proposed AO-PSO algorithm at three SNR levels ($10$, $20$, and $30$~dB). The $x$-axis starts at $k=0$, which represents the ergodic capacity at the initial uniform-spacing positions before any optimization. With weakened PSO parameters ($Z=10$ particles, $I_{\rm PSO}=20$ iterations per AO step), the algorithm exhibits gradual improvement over approximately $5$--$8$ outer AO iterations before stabilizing, confirming the monotonic convergence guaranteed by Proposition~\ref{prop:highSNR}. Moreover, higher SNR yields a larger absolute gain from optimization because the capacity is more sensitive to the eigenvalue distribution of $\mathbf{R}_T$ and $\mathbf{R}_R$ at high SNR.}\label{fig:convergence}
\vspace{-3mm}
\end{figure}

Fig.~\ref{fig:convergence} shows the convergence performance with a reduced PSO budget ($Z = 10$, $I_{\rm PSO} = 20$). As we can see, the capacity is monotonically non-decreasing (Proposition~\ref{prop:convergence}), with the first $3$--$5$ iterations capturing most of the gain. The gain widens with SNR ($\sim$2~bps/Hz at 10~dB, $>$7~bps/Hz at 30~dB), consistent with the result predicted by Propositions~\ref{prop:highSNR} and~\ref{prop:lowSNR}. With default parameters, convergence is achieved in $3$--$5$ iterations.

\vspace{-2mm}
\subsection{Ergodic Capacity vs.\ SNR}

\begin{figure}[t]
\centering
\includegraphics[width=.85\columnwidth]{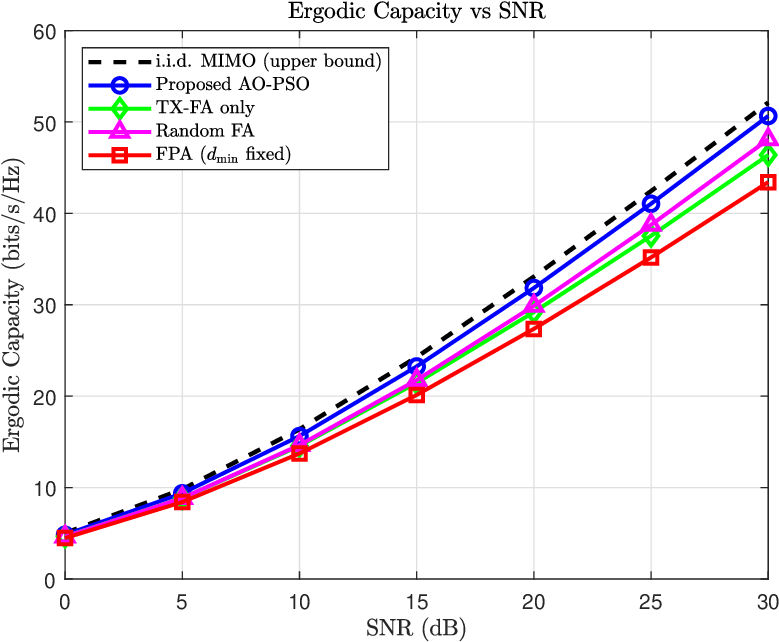}
\caption{Ergodic capacity versus SNR for five schemes: the i.i.d.\ MIMO upper bound ($\mathbf{R}=\mathbf{I}$), the proposed AO-PSO with joint TX+RX FAS optimization, TX-FAS only (RX fixed at $d_{\min}$ spacing), random FAS (best of $50$ random feasible placements), and FPA (both TX and RX fixed at $d_{\min}$ spacing). The proposed AO-PSO consistently outperforms all other practical schemes, with a widening performance gap at higher SNR.}\label{fig:snr}
\vspace{-3mm}
\end{figure}

Fig.~\ref{fig:snr} compares the ergodic capacity of five schemes across SNR $\in[0, 30]$~dB: (i)~\emph{the i.i.d.~MIMO upper bound}: $\mathbf{R}_T=\mathbf{R}_R=\mathbf{I}$, representing the ideal case of fully independent channels; (ii)~\emph{AO-PSO}: both TX and RX positions jointly optimized (proposed); (iii)~\emph{TX-FAS only}: only the TX antennas are optimized via PSO, with the RX fixed at uniform $d_{\min}$ spacing; (iv)~\emph{random FAS}: the best placement among $50$ random feasible configurations for both TX and RX; (v)~\emph{FPA}: both TX and RX fixed at uniform $d_{\min}$ spacing.

The results in this figure indicate a consistent performance hierarchy: $C_{\rm iid} > C_{\rm AO\text{-}PSO} > C_{\rm TX\text{-}FAS} \geq C_{\rm Random} > C_{\rm FPA}$.

\emph{(i)~AO-PSO vs.\ i.i.d.\ upper bound:}
The proposed AO-PSO method closely tracks the i.i.d.\ MIMO bound, with a gap of only $\sim$0.5~bps/Hz at SNR~$= 30$~dB. This residual gap corresponds to the irreducible correlation loss $\Delta C_{\rm AO\text{-}PSO} \approx 1.5$~bps/Hz (Corollary~\ref{cor:loss}), since full decorrelation is unattainable for $N = 6$ within a $2\lambda$ aperture (Remark~\ref{rem:beyond_N2}).

\emph{(ii)~AO-PSO vs.\ FPA:}
AO-PSO outperforms FPA by $\sim$2~bps/Hz at 10~dB and 7+~bps/Hz at 30~dB. The widening gap is predicted by Corollary~\ref{cor:loss}: $\Delta C$ is SNR-independent, so the FPA penalty becomes increasingly dominant at high SNR. At low SNR, all schemes converge as predicted by Proposition~\ref{prop:lowSNR}.

\emph{(iii)~TX-FAS only:}
Optimizing only the TX positions captures roughly half the gain between FPA and AO-PSO, consistent with the multiplicative structure in~\eqref{eq:DeltaC}: optimizing one side eliminates $-\log_2\det(\mathbf{R}_T)$ but retains $-\log_2\det(\mathbf{R}_R)$.

\emph{(iv)~Random FAS:}
The best of $50$ random configurations performs comparably to TX-FAS only, showing that unstructured search cannot match systematic optimization.

\vspace{-2mm}
\subsection{Ergodic Capacity vs.\ Aperture Size}

\begin{figure}[t]
\centering
\includegraphics[width=.85\columnwidth]{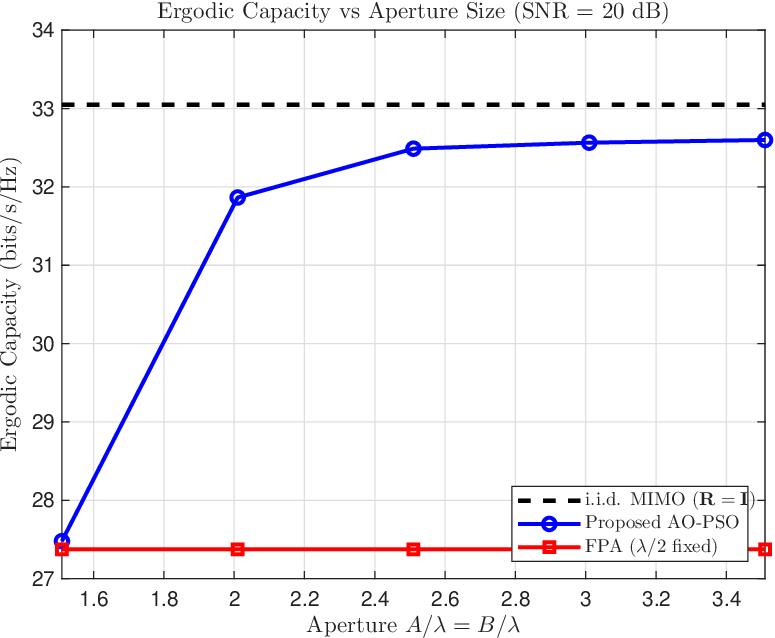}
\caption{Ergodic capacity versus normalized aperture size $A/\lambda=B/\lambda$ at SNR $=20$~dB. The i.i.d.\ MIMO upper bound (dashed) is the capacity achieved when $\mathbf{R}_T=\mathbf{R}_R=\mathbf{I}$. As the aperture grows, both AO-PSO and FPA improve, but AO-PSO achieves the upper bound more rapidly because it can spread the antennas to minimize spatial correlation.}\label{fig:aperture}
\vspace{-3mm}
\end{figure}

Fig.~\ref{fig:aperture} shows the capacity at SNR $= 20$~dB as $A/\lambda = B/\lambda$ increases from $1.5$ to $3.5$.

\emph{(i)~Small aperture ($A/\lambda \leq 1.5$):}
At the minimum feasible aperture, all antennas are packed at $d_{\min}$ with no room for optimization; AO-PSO and FPA perform identically.

\emph{(ii)~Moderate aperture ($1.5 < A/\lambda \leq 3$):}
AO-PSO rapidly approaches $C_{\rm iid}$ (within $1$--$2$~bps/Hz at $A/\lambda = 2$), while FPA remains at a constant low capacity since it does not exploit the available space. The gap exceeds $5$~bps/Hz at $A/\lambda = 2$.

\emph{(iii)~Large aperture ($A/\lambda > 3$):}
The capacity of the AO-PSO method saturates near $C_{\rm iid}$, as predicted by Proposition~\ref{prop:aperture}. Increasing the aperture beyond $3\lambda$ yields diminishing returns for $N = 6$, providing a practical design guideline.

\vspace{-2mm}
\subsection{Optimized Antenna Positions}

\begin{figure}[t]
\centering
\includegraphics[width=.85\columnwidth]{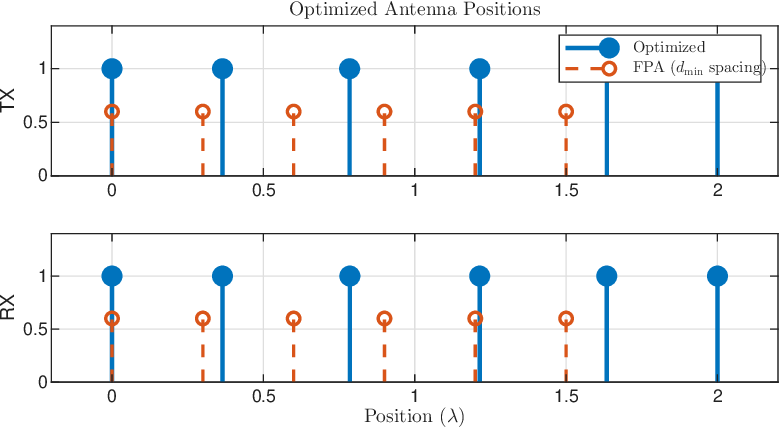}
\caption{Optimized antenna positions (blue filled markers) versus uniform $d_{\min}=0.3\lambda$ spacing (orange dashed markers) at SNR $=20$~dB, $A/\lambda=B/\lambda=2$. The AO-PSO algorithm spreads antennas across the full aperture to maximize inter-element spacing and reduce spatial correlation.}\label{fig:positions}
\vspace{-3mm}
\end{figure}

Fig.~\ref{fig:positions} visualizes the optimized TX and RX antenna positions at SNR $= 20$~dB. Compared to FPA (occupying only $1.5\lambda$ of the $2\lambda$ aperture), the optimized positions spread across the full aperture with inter-element spacing of roughly $0.4\lambda$.

\emph{(i)~Near-uniform spreading:}
The optimized spacings are approximately $|t_{i+1}^* - t_i^*| \approx A\lambda/(N-1) = 0.4\lambda$, close to the first Bessel zero $d^* \approx 0.383\lambda$ (Proposition~\ref{prop:N2}), confirming that the algorithm minimizes the dominant pairwise correlations.

\emph{(ii)~TX--RX symmetry:}
The TX and RX positions are nearly identical ($\mathbf{t}^* \approx \mathbf{r}^*$), as expected from Remark~\ref{rem:decoupled}.

\emph{(iii)~Correlation matrix quality:}
For FPA, $\det(\mathbf{R}_T) \approx 0.015$ and $\kappa > 100$ (severe ill-conditioning). After optimization, $\det(\mathbf{R}_T) \approx 0.59$ and $\kappa < 5$, representing a dramatic improvement in spatial DoF that directly translates to the capacity gains in Figs.~\ref{fig:snr} and~\ref{fig:analytical} via Corollary~\ref{cor:loss}.

\vspace{-2mm}
\subsection{Verification of High-SNR Approximation ($N\!=\!M\!=\!2$)}

\begin{figure}[t]
\centering
\includegraphics[width=.85\columnwidth]{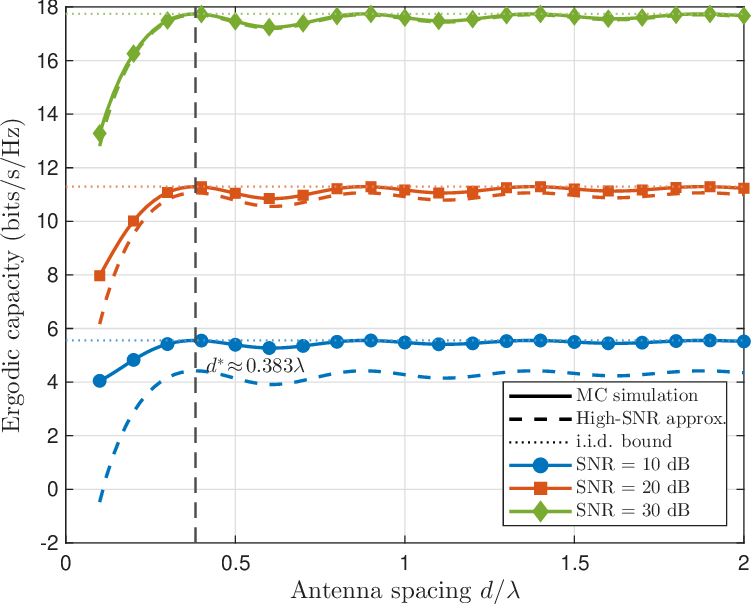}
\caption{Ergodic capacity versus antenna spacing $d/\lambda$ for $N=M=2$ at SNR $\in\{10, 20, 30\}$~dB. Solid lines with markers: Monte-Carlo simulation ($S=3000$). Dashed lines: high-SNR approximation from Proposition~\ref{prop:highSNR}. Dotted horizontal lines: i.i.d.\ MIMO upper bound. The vertical dashed line marks the optimal spacing $d^*\!\approx\!0.383\lambda$ from Proposition~\ref{prop:N2}.}\label{fig:spacing}
\vspace{-3mm}
\end{figure}

Fig.~\ref{fig:spacing} verifies Propositions~\ref{prop:highSNR} and~\ref{prop:N2} for $N\!=\!M\!=\!2$.

\emph{(i)~High-SNR approximation accuracy:}
At SNR~$= 30$~dB, the analytical curve from~\eqref{eq:C_highSNR} closely matches the Monte-Carlo result, confirming tightness. At SNR$=10$~dB, a noticeable gap appears since $\gamma \gg 1$ is less accurate.

\emph{(ii)~Optimal spacing:}
The capacity peaks sharply at $d^* \approx 0.383\lambda$, where $\mathbf{R}_T = \mathbf{R}_R = \mathbf{I}_2$ (zero correlation) and the capacity coincides with the i.i.d.\ bound (Proposition~\ref{prop:N2}).

\emph{(iii)~Oscillatory behavior:}
Beyond $d^*$, the capacity oscillates due to $J_0(\cdot)$, reaching secondary peaks at subsequent zeros. This landscape underscores the non-convexity of \eqref{eq:P0}.

\vspace{-2mm}
\subsection{Verification of Capacity Loss Characterization ($N\!=\!M\!=\!6$)}

\begin{figure}[t]
\centering
\includegraphics[width=.85\columnwidth]{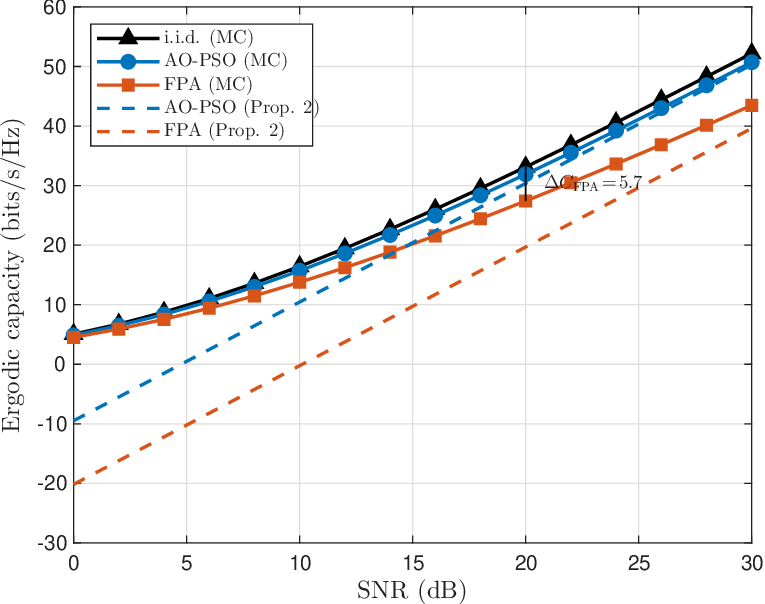}
\caption{Ergodic capacity versus SNR for $N=M=6$, $A/\lambda=B/\lambda=2$. Solid lines with markers: Monte-Carlo simulation. Dashed lines: high-SNR approximation~\eqref{eq:C_highSNR}. The AO-PSO optimized scheme achieves $\det(\mathbf{R}_T)\!\approx\!0.59$, reducing the correlation-induced loss to $\Delta C\!\approx\!1.5$~bps/Hz, whereas FPA has $\det(\mathbf{R}_T)\!\approx\!0.015$ with $\Delta C\!\approx\!12.2$~bps/Hz.}\label{fig:analytical}
\vspace{-3mm}
\end{figure}

Fig.~\ref{fig:analytical} validates Corollary~\ref{cor:loss} for the $N\!=\!M\!=\!6$ setting. For AO-PSO ($\det(\mathbf{R}_T) \approx 0.59$), the high-SNR approximation becomes tight above SNR~$\approx 15$~dB, confirming that the well-conditioned correlation matrices reach the high-SNR
regime at moderate SNR. For FPA ($\det(\mathbf{R}_T) \approx 0.015$), the ill-conditioned matrices require SNR~$> 25$~dB for convergence.

Corollary~\ref{cor:loss} predicts that $\Delta C_{\rm FPA} \approx 12.2$~bps/Hz and $\Delta C_{\rm AO\text{-}PSO} \approx 1.5$~bps/Hz. At SNR~$= 30$~dB, the Monte-Carlo results illustrate that $C_{\rm iid} - C_{\rm FPA} \approx 7.2$~bps/Hz which is smaller than the asymptotic prediction due to finite-SNR corrections, while for the AO-PSO method, the gap closely approaches the predicted $1.5$~bps/Hz, confirming that the algorithm effectively maximizes $\det(\mathbf{R}_T)\det(\mathbf{R}_R)$.

\vspace{-2mm}
\subsection{Validation of Low-SNR and High-SNR Approximations}

\begin{figure}[t]
\centering
\includegraphics[width=.85\columnwidth]{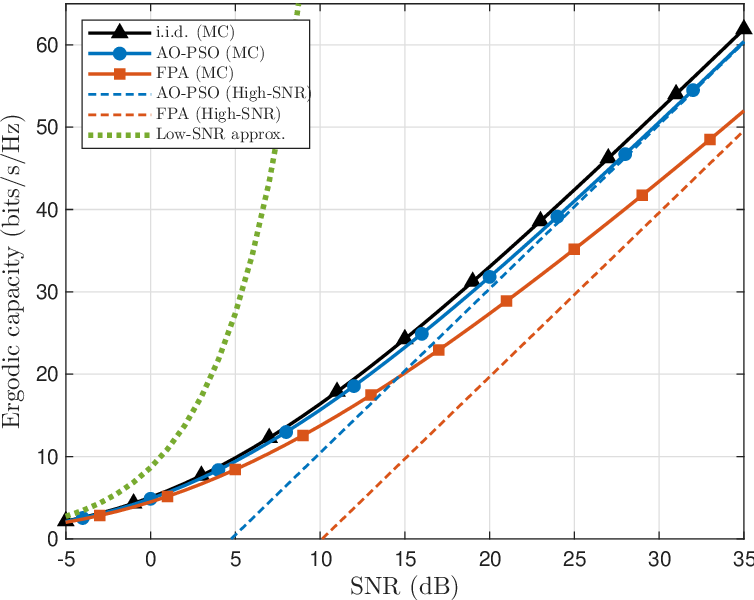}
\caption{Ergodic capacity versus SNR for $N=M=6$, $A/\lambda=B/\lambda=2$. Solid lines with markers: Monte-Carlo simulation. Dashed lines: high-SNR approximation~\eqref{eq:C_highSNR} (Proposition~\ref{prop:highSNR}). Dotted green line: low-SNR approximation~\eqref{eq:C_lowSNR} (Proposition~\ref{prop:lowSNR}). The low-SNR curve is universal across all schemes, while the high-SNR curves separate according to $\det(\mathbf{R}_T)\det(\mathbf{R}_R)$.}\label{fig:lowhigh}
\vspace{-3mm}
\end{figure}

Fig.~\ref{fig:lowhigh} provides a comprehensive validation of both approximations across the full SNR range from $-5$ to $35$~dB.

\emph{(i)~Low-SNR regime ($\text{SNR} < 5$~dB):}
The low-SNR curve $C \approx NM\gamma/\ln 2$ from Proposition~\ref{prop:lowSNR} tightly matches all three Monte-Carlo curves, confirming that the capacity is insensitive to the antenna positions. This regime is power-limited, and FAS optimization provides no noticeable benefit.

\emph{(ii)~High-SNR regime ($\text{SNR} > 15$~dB):}
The high-SNR approximation~\eqref{eq:C_highSNR} converges to the Monte-Carlo results. For AO-PSO ($\det(\mathbf{R}_T) \approx 0.59$), the approximation becomes accurate above $\sim$15~dB. For FPA ($\det(\mathbf{R}_T) \approx 0.015$), the ill-conditioned matrices require SNR~$> 25$~dB for convergence.

\emph{(iii)~Transition region ($5$--$15$~dB):}
Neither approximation is tight; the curves begin to separate as the system transitions from power-limited to DoF-limited operation.

\vspace{-2mm}
\subsection{Multi-Antenna Scaling ($N = M \in \{2,\ldots,8\}$)}

\begin{figure}[t]
\centering
\includegraphics[width=.65\columnwidth]{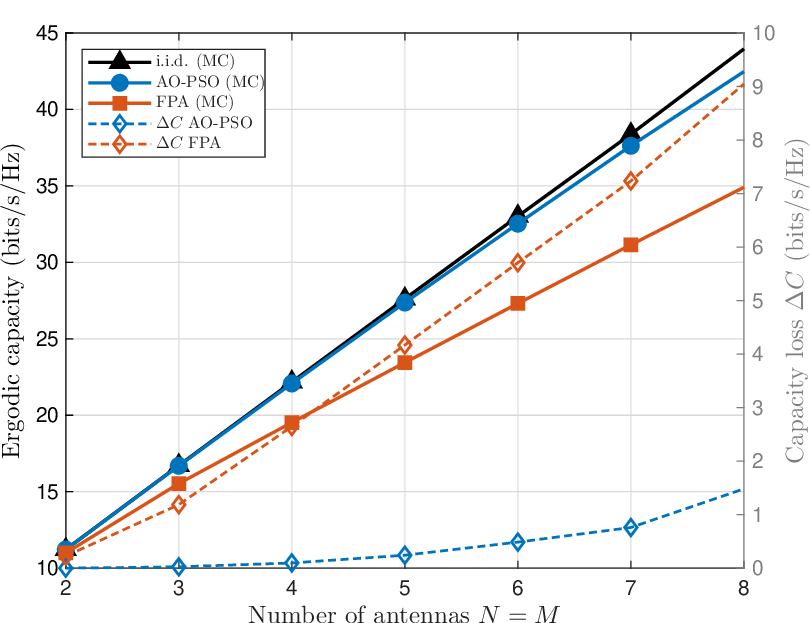}
\caption{Ergodic capacity (left axis, solid lines) and capacity loss $\Delta C = C_{\rm iid} - C$ (right axis, dashed lines) versus the number of fluid antennas $N = M$ at SNR $= 20$~dB with aperture $A/\lambda = B/\lambda = 3$. The AO-PSO scheme maintains near-i.i.d.\ MIMO capacity across all $N$, while the FPA capacity loss grows rapidly with $N$.}\label{fig:scaling}
\vspace{-3mm}
\end{figure}

Fig.~\ref{fig:scaling} investigates the scalability for $N = M \in \{2,\ldots,8\}$ with $A/\lambda = 3$ and SNR~$= 20$~dB.

\emph{(i)~Linear capacity growth:}
All schemes exhibit approximately linear growth with $N$ (from the $N\log_2\gamma$ term in~\eqref{eq:C_highSNR}), but the AO-PSO curve closely tracks $C_{\rm iid}$ while FPA grows more slowly. At $N = 8$, AO-PSO achieves $42.5$~bps/Hz versus $34.9$~bps/Hz for FPA---a gap of $7.6$~bps/Hz.

\emph{(ii)~Capacity loss scaling:}
The FPA loss $\Delta C_{\rm FPA}$ grows rapidly: from $0.2$~bps/Hz ($N = 2$) to $9.1$~bps/Hz ($N = 8$), driven by the exponential decay of $\det(\mathbf{R}_T)$ ($\approx 10^{-4}$ at $N = 8$). In contrast, AO-PSO maintains $\Delta C < 1.5$~bps/Hz and $\det(\mathbf{R}_T) > 0.5$ even at $N = 8$. For $N = 2$, perfect decorrelation is achieved (Proposition~\ref{prop:N2}).

\emph{(iii)~Practical implication:}
FAS optimization is most critical for large $N$, where the FPA loss becomes prohibitive. A rule of thumb: $A/\lambda \geq 0.5N$ for near-optimal performance.

\vspace{-2mm}
\subsection{Convergence Comparison: AO-PSO vs.\ AO-SCA}

\begin{figure}[t]
\centering
\includegraphics[width=.65\columnwidth]{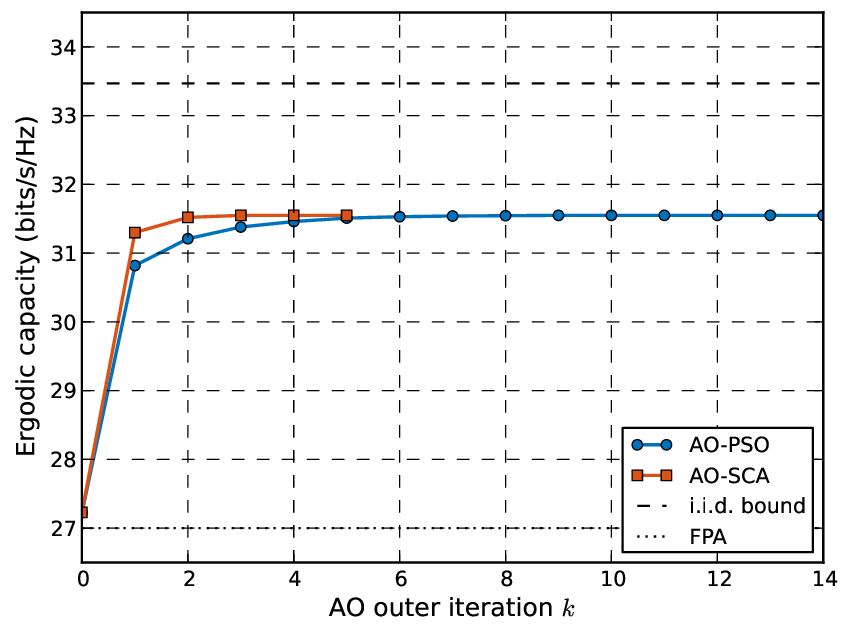}
\caption{Convergence comparison of AO-PSO ($Z\!=\!15$, $I_{\rm PSO}\!=\!40$) and AO-SCA ($I_{\rm SCA}\!=\!50$, $\eta_0\!=\!0.02\lambda$) at SNR~$=20$~dB with $N\!=\!M\!=\!6$. All capacity values are re-evaluated with $S\!=\!1000$ fixed Monte-Carlo samples. Horizontal lines: i.i.d.~MIMO upper bound (dashed) and FPA baseline (dotted).}\label{fig:conv_compare}
\vspace{-3mm}
\end{figure}

Fig.~\ref{fig:conv_compare} compares the convergence trajectories of the AO-PSO and AO-SCA methods at SNR~$= 20$~dB. All capacity values are evaluated with a fixed set of $S = 1000$ Monte-Carlo samples to ensure a fair and smooth comparison.

\emph{(i)~AO-SCA convergence speed:}
AO-SCA converges in only $2$--$3$ AO outer iterations, with the majority of the gain in the first iteration. This is enabled by the closed-form gradient~\eqref{eq:grad_fn}, which provides an exact ascent direction without noise.

\emph{(ii)~AO-PSO convergence behavior:}
AO-PSO converges over $5$--$8$ iterations. With default parameters ($Z = 20$, $I_{\rm PSO} = 60$), convergence
is achieved in $3$--$5$ iterations.

\emph{(iii)~Identical converged solutions:}
Both algorithms converge to $\det(\mathbf{R}_T) \approx 0.587$, confirming that the landscape has a single dominant basin of attraction. This validates Remark~\ref{rem:local_opt}: despite being a local optimizer, AO-SCA finds the same solution as the global PSO search.

\vspace{-2mm}
\subsection{Capacity Comparison: AO-PSO vs.\ AO-SCA vs.\ FPA}

\begin{figure}[t]
\centering
\includegraphics[width=.65\columnwidth]{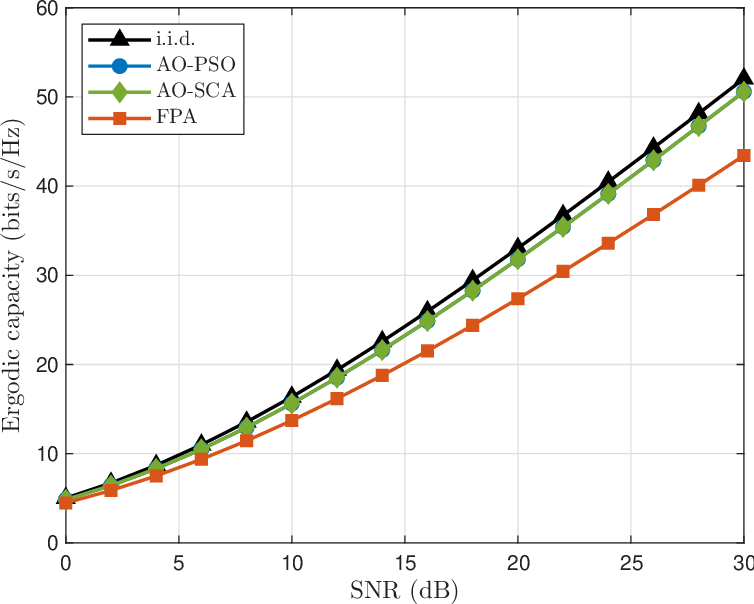}
\caption{Ergodic capacity versus SNR for $N=M=6$, $A/\lambda=B/\lambda=2$. Comparison of AO-PSO, AO-SCA (optimizing $\log_2\det(\mathbf{R}_T)$), FPA, and i.i.d.~MIMO upper bound. The AO-SCA and AO-PSO curves are virtually indistinguishable across the entire SNR range.}\label{fig:pso_vs_sca}
\vspace{-3mm}
\end{figure}

Fig.~\ref{fig:pso_vs_sca} compares the ergodic capacity of AO-PSO, AO-SCA, and FPA across SNR $\in [0, 30]$~dB.

\emph{(i)~AO-PSO vs.\ AO-SCA:}
Both are virtually indistinguishable ($<0.1$~bps/Hz gap), showing that the high-SNR surrogate $\log_2\det(\mathbf{R}_T)$ is a good proxy for ergodic capacity even at moderate SNR. Since the optimal positions depend on the SNR-independent correlation structure, positions optimized at the high-SNR limit remain near-optimal at all SNR levels.

\emph{(ii)~Computational advantage:}
AO-SCA completes in $<0.1$~s on a modern workstation versus $3$--$5$~min for AO-PSO, while achieving the same capacity.

\emph{(iii)~AO-PSO \& AO-SCA vs.\ FPA:}
Both outperform FPA by $>$7~bps/Hz at SNR$=30$~dB, consistent with Corollary~\ref{cor:loss}. The gap widens with SNR, confirming that position optimization is most beneficial in the high-SNR regime.

\emph{(iv)~Practical recommendation:}
Given the near-identical capacity and vastly superior computational efficiency, AO-SCA is recommended as the primary optimization method. AO-PSO serves as a complementary tool for scenarios with poor initialization or many local optima.

\vspace{-2mm}
\subsection{Summary of Performance and Complexity}

Table~\ref{tab:comparison} consolidates the key performance and complexity metrics for all schemes under the default setting ($N\!=\!M\!=\!6$, $A/\lambda\!=\!B/\lambda\!=\!2$, SNR~$=30$~dB). The proposed AO-PSO and AO-SCA both achieve $\det(\mathbf{R}_T)\approx 0.59$, reducing the capacity gap to the i.i.d.~MIMO bound to merely $\sim$0.5~bps/Hz---an order-of-magnitude improvement over FPA ($\sim$7.2~bps/Hz gap). The AO-SCA variant delivers the same decorrelation quality with a runtime speedup exceeding $10^5$ times.

\begin{table}[t]
\centering
\caption{Performance and Complexity Comparison with $N=M=6$, $A/\lambda=B/\lambda=2$, ${\rm SNR}=30~{\rm dB}$}\label{tab:comparison}
\resizebox{.65\columnwidth}{!}{
\begin{tabular}{lcccc}
\hline
\textbf{Scheme} & $\det(\mathbf{R}_T)$ & \textbf{Gap to i.i.d.} & \textbf{AO iter.} & \textbf{Runtime}\\
 & & (bps/Hz) & & \\
\hline
i.i.d.\ bound   & 1.000 & 0       & ---   & ---       \\
AO-PSO (proposed)          & 0.587 & $\sim$0.5   & 5--8  & 3--5 min  \\
AO-SCA (proposed)          & 0.587 & $\sim$0.5   & 2--3  & $<$0.1 s  \\
TX-FAS only      & 0.587$^\dagger$ & $\sim$3.9   & 5--8  & 1--2 min  \\
Random FAS       & varies & $\sim$4.0   & ---   & negligible\\
FPA              & 0.015 & $\sim$7.2   & ---   & ---       \\
\hline
\end{tabular}}
\begin{tablenotes}
\item {\scriptsize $^\dagger$TX side only; RX remains at $\det(\mathbf{R}_R)\!=\!0.015$ (FPA).}
\end{tablenotes}
\vspace{-3mm}
\end{table}
 
\section{Conclusion}\label{sec:conclusion}
This paper studied spatial-correlation optimization for fluid MIMO systems, where both the TX and RX adopt multiple fluid antennas over a linear aperture. Under Jakes' model and Kronecker channel decomposition, we derived the position-dependent correlation matrices and formulated an ergodic-capacity maximization problem. Our analysis yielded a high-SNR approximation that makes the role of $\det(\mathbf{R}_T)\det(\mathbf{R}_R)$ explicit, a closed-form capacity-loss bound, and a globally optimal spacing rule for the two-antenna case. For the general case, we developed two solvers: AO-PSO for global exploration and AO-SCA, which leverages the gradient via $J_0'(x)=-J_1(x)$ to solve a sequence of convex surrogates. Simulation results showed sizable gains over FPA arrays and optimized fluid MIMO achieves near-i.i.d.~capacity performance.


\begin{thebibliography}{99}
\bibitem{Paulraj-1994}
A. J Paulraj, Thomas Kailath, ``Increasing capacity in wireless broadcast systems using distributed transmission/directional reception (DTDR),'' US Patent 5,345,599, 1994.
\bibitem{Foschini-1996}
G. J. Foschini, ``Layered space-time architecture for wireless communication in a fading environment when using multi-element antennas,'' {\em Bell Labs Tech. J.}, vol. 1, no. 2, pp. 41--59, Autumn 1996.
\bibitem{Telatar-1999}
E. Telatar, ``Capacity of multi-antenna Gaussian channels,'' {\em Euro. Trans. Telecommun.}, vol. 10, no. 6, pp. 585--595, 1999.
\bibitem{Tarokh-1998}
V. Tarokh, N. Seshadri and A. R. Calderbank, ``Space-time codes for high data rate wireless communication: Performance criterion and code construction,'' {\em IEEE Trans. Inf. Theory}, vol. 44, no. 2, pp. 744--765, Mar. 1998.
\bibitem{Tarokh-1999}
V. Tarokh, H. Jafarkhani and A. R. Calderbank, ``Space-time block codes from orthogonal designs,'' {\em IEEE Trans. Inf. Theory}, vol. 45, no. 5, pp. 1456--1467, Jul. 1999.
\bibitem{Raleigh-1998}
G. G. Raleigh and J. M. Cioffi, ``Spatio-temporal coding for wireless communication,'' {\em IEEE Trans. Commun.}, vol. 46, no. 3, pp. 357--366, Mar. 1998.

\bibitem{wong2000opt}
K. K. Wong, R. D. Murch, R. S.-K. Cheng and K. B. Letaief, ``Optimizing the spectral efficiency of multiuser MIMO smart antenna systems,'' in \emph{Proc. IEEE Wirel. Commun. Netw. Conf. (WCNC)}, vol. 1, pp. 426--430, 23-28 Sept. 2000, Chicago, IL, USA.
\bibitem{wong2002per}
K. K. Wong, R. D. Murch and K. B. Letaief, ``Performance enhancement of multiuser MIMO wireless communication systems,"  \emph{IEEE Trans. Commun.}, vol. 50, no. 12, pp. 1960--1970, Dec. 2002.
\bibitem{wong2003jcd}
K. K. Wong, R. D. Murch, and K. Ben Letaief, ``A joint-channel diagonalization for multiuser MIMO antenna systems,'' {\em IEEE Trans. Wireless Commun.}, vol. 4, no. 2, pp. 773--786, Jul. 2003.
\bibitem{Vishwanath-2003}
S. Vishwanath, N. Jindal and A. Goldsmith, ``Duality, achievable rates, and sum-rate capacity of Gaussian MIMO broadcast channels,'' {\em IEEE Trans. Inf. Theory}, vol. 49, no. 10, pp. 2658--2668, Oct. 2003.
\bibitem{Choi-2004}
Lai-U Choi and R. D. Murch, ``A transmit preprocessing technique for multiuser MIMO systems using a decomposition approach,'' {\em IEEE Trans. Wireless Commun.}, vol. 3, no. 1, pp. 20--24, Jan. 2004.
\bibitem{Spencer-2004}
Q. H. Spencer, A. L. Swindlehurst and M. Haardt, ``Zero-forcing methods for downlink spatial multiplexing in multiuser MIMO channels,'' {\em IEEE Trans. Sig. Process.}, vol. 52, no. 2, pp. 461--471, Feb. 2004.
\bibitem{Villalonga2022spectral}
D.~A.~U. Villalonga, H.~OdetAlla, M.~J. Fern\'{a}ndez-Getino Garc\'{i}a, and A.~Flizikowski, ``Spectral efficiency of precoded 5G-NR in single and multi-user scenarios under imperfect channel knowledge: A comprehensive guide for implementation,'' {\em Electronics}, vol. 11, no. 24, p. 4237, Dec. 2022.
\bibitem{10379539}
Z.~Wang {\em et al.}, ``A tutorial on extremely large-scale {MIMO} for {6G}: Fundamentals, signal processing, and applications,'' {\em IEEE Commun. Surv. \& Tut.}, vol. 26, no. 3, pp. 1560--1605, thirdquarter 2024.
\bibitem{Hung-2014}
Y.-C. Hung and S.-H. L. Tsai, ``PAPR analysis and mitigation algorithms for beamforming MIMO OFDM systems,'' {\em IEEE Trans. Wireless Commun.}, vol. 13, no. 5, pp. 2588--2600, May 2014.

\bibitem{Tariq-2020}
F. Tariq {\em et al.}, ``A speculative study on 6G,'' {\em IEEE Wireless Commun.}, vol. 27, no. 4, pp. 118--125, Aug. 2020.
\bibitem{Bernhard-2007}
J. T. Bernhard, {\em Reconfigurable antennas}, Synthesis Lectures on Antennas, 2007.

\bibitem{wang2021fluid}
K. K. Wong, A. Shojaeifard, K.-F. Tong and Y. Zhang, ``Performance limits of fluid antenna systems," \emph{IEEE Commun. Lett.}, vol. 24, no. 11, pp. 2469--2472, Nov. 2020.
\bibitem{wong2020perf}
K. K. Wong, A. Shojaeifard, K.-F. Tong and Y. Zhang, ``Fluid antenna systems," \emph{IEEE Trans. Wireless Commun.}, vol. 20, no. 3, pp. 1950--1962, Mar. 2021.
\bibitem{new2025tut}
W. K. New {\em et al.}, ``A tutorial on fluid antenna system for 6G networks: Encompassing communication theory, optimization methods and hardware designs,'' \emph{IEEE Commun. Surv. Tuts.}, vol. 27, no. 4, pp. 2325--2377, Aug. 2025.
\bibitem{Lu-2025}
W.-J. Lu {\em et al.}, ``Fluid antennas: Reshaping intrinsic properties for flexible radiation characteristics in intelligent wireless networks,'' {\em IEEE Commun. Mag.}, vol. 63, no. 5, pp. 40--45, May 2025.
\bibitem{hong2025contemporary}
H. Hong {\em et al.}, ``A contemporary survey on fluid antenna systems: Fundamentals and networking perspectives,'' {\em  IEEE Trans. Netw. Sci. Eng.},  vol. 13, pp. 2305--2328, 2026. 
\bibitem{new2025flar}
W. K. New \emph{et al.}, ``Fluid antenna systems: Redefining reconfigurable wireless communications," \emph{IEEE J. Sel. Areas Commun.}, vol. 44, pp. 1013--1044, 2026.
\bibitem{wu2024flu}
T. Wu {\em et al.}, ``Fluid antenna systems enabling 6G: Principles, applications, and research directions,'' to appear in \emph{IEEE Wireless Commun.}, \url{DOI: 10.1109/MWC.2025.3629597}, 2025.
\bibitem{Zhu-Wong-2024}
L. Zhu and K. K. Wong, ``Historical review of fluid antennas and movable antennas,'' {\em arXiv preprint}, \url{arXiv:2401.02362v2}, 2024.
\bibitem{shen2024design}
Y. Shen {\em et al.}, ``Design and experimental validation of mmWave surface wave enabled fluid antennas for future wireless communications,'' to appear in {\em IEEE Antennas \& Wireless Propag. Lett.}, \url{DOI: 10.1109/LAWP.2026.3657059}, 2026.
\bibitem{Shamim-2025}
R. Wang {\em et al.}, ``Electromagnetically reconfigurable fluid antenna system for wireless communications: Design, modeling, algorithm, fabrication, and experiment,'' {\em IEEE J. Select. Areas Commun.}, vol. 44, pp. 1464--1479, 2026.
\bibitem{zhang2024pixel}
J. Zhang {\em et al.}, ``A novel pixel-based reconfigurable antenna applied in fluid antenna systems with high switching speed,'' {\em IEEE Open J. Antennas \& Propag.}, vol. 6, no. 1, pp. 212--228, Feb. 2025.
\bibitem{tong-2025pixel}
B. Liu, T. Wu, K. K. Wong, H. Wong, and K. F. Tong, ``Wideband pixel-based fluid antenna system: An antenna design for smart city,'' {\em IEEE Internet Things J.}, vol. 13, no. 4, pp. 6850--6862, Feb. 2026.
\bibitem{Wong-wc2026}
K. K. Wong, C. Wang, S. Shen, C.-B. Chae and R. Murch, ``Reconfigurable pixel antennas meet fluid antenna systems: A paradigm shift to electromagnetic signal and information processing,'' {\em IEEE Wireless Commun.}, vol. 33, no. 1, pp. 191--198, Feb. 2026.
\bibitem{Zhang-jsac2026}
S. Zhang {\em et al.}, ``Fluid antenna systems enabled by reconfigurable holographic surfaces: Beamforming design and experimental validation,'' {\em IEEE J. Select. Areas Commun.}, vol. 44, pp. 1417--1431, 2026.
\bibitem{Liu-2025arxiv}
B. Liu, K.-F. Tong, K. K. Wong, C.-B. Chae, and H. Wong, ``Programmable meta-fluid antenna for spatial multiplexing in fast fluctuating radio channels,'' {\em Optics Express}, vol. 33, no. 13, pp. 28898--28915, 2025.
\bibitem{tong2025design}
K. F. Tong, B. Liu, and K. K. Wong. ``Designs and challenges in fluid antenna system hardware.'' {\em Electronics}, vol~14, no~7, pp.~1458, 2025.

\bibitem{FAS22}
K. K. Wong, K.-F. Tong, Y. Chen, and Y. Zhang, ``Closed-form expressions for spatial correlation parameters for performance analysis of fluid antenna systems,'' \emph{IET Electron. Lett.}, vol.~58, no.~11, pp.~454--457, May~2022.
\bibitem{khammassi2023new}
M. Khammassi, A. Kammoun, and M.-S. Alouini, ``A new analytical approximation of the fluid antenna system channel,'' {\em IEEE Trans. Wireless Commun.}, vol. 22, no. 12, pp. 8843--8858, Dec. 2023.
\bibitem{new2024fluid}
W. K. New, K. K. Wong, H. Xu, K.-F. Tong and C.-B. Chae, ``Fluid antenna system: New insights on outage probability and diversity gain," \emph{IEEE Trans. Wireless Commun.}, vol. 23, no. 1, pp. 128--140, Jan. 2024.
\bibitem{new2024an}
W. K. New, K. K. Wong, H. Xu, K.-F. Tong and C.-B. Chae, ``An information-theoretic characterization of MIMO-FAS: Optimization, diversity-multiplexing tradeoff and $q$-outage capacity,'' {\em IEEE Trans. Wireless Commun.}, vol. 23, no. 6, pp. 5541--5556, Jun. 2024.
\bibitem{Ghadi-2023}
F. R. Ghadi {\em et al.}, ``Copula-based performance analysis for fluid antenna systems under arbitrary fading channels,'' \emph{IEEE Commun. Lett.}, vol.~27, no.~11, pp.~3068--3072, Nov.~2023.
\bibitem{10678877}
F. Rostami Ghadi {\em et al.}, ``A Gaussian copula approach to the performance analysis of fluid antenna systems,'' \emph{IEEE Trans. Wireless Commun.}, vol. 23, no. 11, pp. 17573--17585, Nov. 2024.
\bibitem{ramirez2024new}
P. Ram\'{i}rez-Espinosa, D. Morales-Jimenez and K. K. Wong, ``A new spatial block-correlation model for fluid antenna systems,'' {\em IEEE Trans. Wireless Commun.}, vol.~23, no.~11, pp. 15829--15843, Nov. 2024.
\bibitem{LaiX}
X. Lai {\em et al.}, ``Revisiting spatial block-correlation model for fluid antenna systems: From constant to variable correlations,'' \emph{IEEE J. Sel. Areas Commun.}, vol.~44, pp.~1335--1351, 2026.

\bibitem{G3_chai2022SISO-FAS-PS}
Z. Chai, K. K. Wong, K. F. Tong, Y. Chen and Y. Zhang, ``Port selection for fluid antenna systems,'' {\em IEEE Commun. Lett.}, vol. 26, no. 5, pp. 1180--1184, May 2022.
\bibitem{XLai23}
X. Lai {\em et al.}, ``On performance of fluid antenna system using maximum ratio combining,'' \emph{IEEE Commun. Lett.}, vol.~28, no.~2, pp.~402--406, Feb.~2024.
\bibitem{10416896}
H.~Qin {\em et al.}, ``Antenna positioning and beamforming design for fluid antenna-assisted multi-user downlink communications,'' \emph{IEEE Wireless Commun. Lett.}, vol.~13, no.~4, pp. 1073--1077, Apr. 2024.
\bibitem{LZhu23}
L. Zhu, W. Ma, B. Ning and R. Zhang, ``Movable-antenna enhanced multiuser communication via antenna position optimization,'' \emph{IEEE Trans. Wireless Commun.}, vol.~23, no.~7, pp.~7214--7229, Jul.~2024.
\bibitem{10767351}
Y.~Chen {\em et al.}, ``Joint beamforming and antenna design for near-field fluid antenna system,'' \emph{IEEE Wireless Commun. Lett.}, vol.~14, no.~2, pp. 415--419, Feb. 2025.

\bibitem{YaoJ241}
J. Yao {\em et al.}, ``FAS-driven spectrum sensing for cognitive radio networks,'' \emph{IEEE Internet Things J.}, vol.~12, no.~5, pp.~6046--6049, Mar.~2025.
\bibitem{wang2024fluid}
C. Wang {\em et al.}, ``Fluid antenna system liberating multiuser MIMO for ISAC via deep renforcement learning,'' {\em IEEE Trans. Wireless Commun.}, vol. 23, no. 9, pp. 10879--10894, Sept. 2024.
\bibitem{zou2024shift}
J. Zou {\em et al.}, ``Shifting the ISAC trade-off with fluid antenna systems,'' \emph{IEEE Wireless Commun. Lett.}, vol.~13, no.~12, pp.~3479--3483, Dec. 2024.
\bibitem{Zhou-isac2024}
L. Zhou, J. Yao, M. Jin, T. Wu and K. K. Wong, ``Fluid antenna-assisted ISAC systems,'' {\em IEEE Wireless Commun. Lett.}, vol. 13, no. 12, pp. 3533--3537, Dec. 2024.
\bibitem{SYangWCL25}
S. Yang {\em et al.}, ``Toward intelligent antenna positioning: Leveraging DRL for FAS-aided ISAC systems,'' \emph{IEEE Internet Things J.}, vol.~14, no.~7, pp.~2029--2033, Jul.~2025.

\bibitem{new-2024noma}
W. K. New {\em et al.}, ``Fluid antenna system enhancing orthogonal and non-orthogonal multiple access,'' {\em IEEE Commun. Lett.}, vol. 28, no. 1, pp. 218--222, Jan. 2024.
\bibitem{YaoJ252}
J. Yao {\em et al.}, ``Exploring fairness for FAS-assisted communication systems: From NOMA to OMA,'' \emph{IEEE Trans. Wireless Commun.}, vol.~24, no.~4, pp.~3433--3449, Apr.~2025.
\bibitem{TWuTCOMM26}
T. Wu {\em et al.}, ``Unleashing more potential from FAS: A framework of FAS-CoNOMA systems,'' \emph{IEEE Trans. Commun.}, vol. 74, pp. 4820--4836, 2026.

\bibitem{Ghadi2024005}
F. Rostami Ghadi {\em et al.}, ``On performance of RIS-aided fluid antenna systems,'' {\em IEEE Wireless Commun. Lett.}, vol. 13, no. 8, pp. 2175--2179, Aug. 2024.
\bibitem{LaiX242}
X. Lai {\em et al.}, ``FAS-RIS: A block-correlation model analysis,'' \emph{IEEE Trans. Veh. Technol.}, vol.~74, no.~2, pp.~3412--3417, Feb.~2025.
\bibitem{YaoJ251}
J. Yao {\em et al.}, ``FAS-RIS communication: Model, analysis, and optimization,'' \emph{IEEE Trans. Veh. Technol.}, vol.~74, no.~6, pp.~9938--9943, Jun.~2025.
\bibitem{YaoJTWC26}
J. Yao {\em et al.}, ``FAS vs. ARIS: Which is more important for FAS-ARIS communication systems?,'' \emph{IEEE Trans. Wireless Commun.}, vol.~25, pp.~2075--2091, 2026.
\bibitem{TWuTVT25}
T. Wu {\em et al.}, ``FAS-RIS for V2X: Unlocking realistic performance analysis with finite elements,'' \emph{IEEE Trans. Veh. Technol.}, \url{DOI: 10.1109/TVT.2025.3647789}, 2025.
\bibitem{HChenTNSE25}
H. Chen {\em et al.}, ``FAS-ARIS: Turning multipath challenges into localization opportunities,'' \emph{IEEE Trans. Netw. Sci. Eng.}, vol. 13, pp. 3756--3772, 2026.

\bibitem{HXuTWC25}
H. Xu {\em et al.}, ``The future is fluid: Revolutionizing DOA estimation with sparse fluid antennas,'' \emph{IEEE Trans. Wireless Commun.}, vol. 25, pp. 11546--11561, 2026.
\bibitem{TWuJSTSP25}
T. Wu {\em et al.}, ``Scalable FAS: A new paradigm for array signal processing,'' to appear in \emph{IEEE J. Sel. Topics Signal Process.}, \url{arXiv:2508.10831}, 2026.

\bibitem{Security3}
B. Tang {\em et al.}, ``Fluid antenna enabling secret communications,'' \emph{IEEE Commun. Lett.}, vol.~27, no.~6, pp.~1491--1495, Jun.~2023.
\bibitem{Security2}
J. D. Vega-S\'{a}nchez, L. F. Urquiza-Aguiar, H. R. C. Mora, N. V. O. Garz\'{o}n, and D. P. M. Osorio, ``Fluid antenna system: Secrecy outage probability analysis,'' \emph{IEEE Trans. Veh. Technol.}, vol.~73, no.~8, pp.~11458--11469, Aug.~2024.
\bibitem{ghadi2024phys}
F. Rostami Ghadi {\em et al.}, ``Physical layer security over fluid antenna systems: Secrecy performance analysis,'' {\em IEEE Trans. Wireless Commun.}, vol. 23, no. 12, pp. 18201--18213, Dec. 2024.
\bibitem{Yao-2025pls}
J. Yao {\em et al.}, ``FAS for secure and covert communications,'' {\em IEEE Internet Things J.}, vol. 12, no. 11, pp. 18414--18418, Jun. 2025.
\bibitem{Security5}
J. Zheng {\em et al.}, ``Unlocking FAS-RIS security analysis with block-correlation model,'' \emph{IEEE Wireless Commun. Lett.}, vol. 14, no. 7, pp. 2029--2033, Jul. 2025.
\bibitem{TuoW}
T. Wu {\em et al.}, ``Variable block-correlation modeling and optimization for secrecy analysis in fluid antenna systems,'' \emph{arXiv preprint}, \url{ arXiv:2510.03594}, 2025.

\bibitem{JYao2024}
J. Yao {\em et al.}, ``Proactive monitoring via jamming in fluid antenna systems,'' \emph{IEEE Commun. Lett.}, vol.~28, no.~7, pp.~1695--1702, Jul.~2024.

\bibitem{H4_wong2022FAMA}
K. K. Wong and K. F. Tong, ``Fluid antenna multiple access,'' \emph{IEEE Trans. Wireless Commun.}, vol.~21, no.~7, pp. 4801--4815, Jul. 2022.
\bibitem{H5_wong2023fast}
K. K. Wong, K. F. Tong, Y. Chen, and Y. Zhang, ``Fast fluid antenna multiple access enabling massive connectivity,'' {\em IEEE Commun. Lett.}, vol. 27, no. 2, pp. 711--715, Feb. 2023.
\bibitem{H6_wong2023sFAMA}
K. K. Wong, D. Morales-Jimenez, K. F. Tong, and C.-B. Chae, ``Slow fluid antenna multiple access,'' \emph{IEEE Trans. Commun.}, vol.~71, no.~5, pp. 2831--2846, May 2023.
\bibitem{Coma-2026fama}
J. P. Gonz\'{a}lez-Coma and F. J. L\'{o}pez-Mart\'{i}nez, ``Slow fluid antenna multiple access with multiport receivers,'' {\em IEEE Wireless Commun. Lett.}, vol. 15, pp. 1280--1284, 2026.
\bibitem{Yuan-2026fama}
X. Yuan, N. Guo, Y. Hu, R. Schober and A. Schmeink, ``Optimal antenna configuration filtering and joint power control in fluid antenna multiple access networks,'' {\em IEEE J. Select. Areas Commun.}, vol. 44, pp. 1227--1242, 2026.
\bibitem{Dinis-2026fama}
D. Dinis and R. Wichman, ``s-FAMA-GP: A low-complexity slow FAMA using interference interpolation,'' {\em IEEE Wireless Commun. Lett.}, vol. 15, pp. 1727--1731, 2026.
\bibitem{Zhang-fama2025}
Z. Zhang, \emph{et al.}, ``On fundamental limits of slow-fluid antenna multiple access for unsourced random access,'' \emph{IEEE Wireless Commun. Lett.}, vol. 14, no. 11, pp. 3455--3459, Nov. 2025.
\bibitem{H12_Wong2024cuma}
K. K. Wong, C. B. Chae, and K. F. Tong, ``Compact ultra massive antenna array: A simple open-loop massive connectivity scheme,'' {\em IEEE Trans. Wireless Commun.}, vol. 23, no. 6, pp. 6279--6294, Jun. 2024.
\bibitem{H10_hong2024coded}
H. Hong, K. K. Wong, K. F. Tong, H. Shin, and Y. Zhang, ``Coded fluid antenna multiple access over fast fading channels,'' \emph{IEEE Wireless Commun. Lett.}, vol.~14, no.~4, pp.~1249--1253, Apr. 2025.
\bibitem{H11_hong2025Downlink}
H. Hong {\em et al.}, ``Downlink OFDM-FAMA in 5G-NR systems,'' {\em IEEE Trans. Wireless Commun.}, vol. 24, no. 12, pp. 10116--10132, Dec. 2025.
\bibitem{Waqar-2025}
N. Waqar, K. K. Wong, C.-B. Chae, and R. Murch, ``Turbocharging fluid antenna multiple access,'' {\em IEEE Trans. Wireless Commun.}, vol. 25, pp. 4038--4052, 2025. 
\bibitem{Waqar-2026wcl}
N. Waqar, K. K. Wong, C.-B. Chae, and R. Murch, ``Attentional copula-aided turbo fluid antenna massive access,'' to appear in {\em IEEE Wireless Commun. Lett.}, \url{DOI: 10.1109/LWC.2026.3665494}, 2026.

\bibitem{xu2024channel}
H. Xu {\em et al.}, ``Channel estimation for FAS-assisted multiuser mmWave systems,'' {\em IEEE Commun. Lett.}, vol. 28, no. 3, pp. 632--636, Mar. 2024.
\bibitem{zhang2024learning}
H. Zhang {\em et al.}, ``Learning-induced channel extrapolation for fluid antenna systems using asymmetric graph masked autoencoder,'' IEEE Wireless Commun. Lett., vol. 13, no. 6, pp. 1665--1669, Jun. 2024.
\bibitem{new2025channel}
W. K. New {\em et al.}, ``Channel estimation and reconstruction in fluid antenna system: Oversampling is essential,'' {\em IEEE Trans. Wireless Commun.}, vol. 24, no. 1, pp. 309--322, Jan. 2025.
\bibitem{10807122}
Z. Zhang, J. Zhu, L. Dai, and R. W. Heath Jr, ``Successive Bayesian reconstructor for channel estimation in fluid antenna systems,'' {\em  IEEE Trans. Wireless Commun.},vol. 24, no. 3, pp. 1992--2006, Mar. 2025.
\bibitem{xu2024sparse}
B.~Xu, Y.~Chen, Q.~Cui, X.~Tao, and K. K. Wong, ``Sparse Bayesian learning-based channel estimation for fluid antenna systems,'' \emph{IEEE Wireless Commun. Lett.}, vol. 14, no. 2, pp. 325--329, Feb. 2025.

\bibitem{Jakes74}
W. C. Jakes, \emph{Microwave Mobile Communications}. New York, NY, USA: Wiley, 1974.
\end{thebibliography}
\end{document}